\documentclass{article}
\usepackage{placeins}
\usepackage{float}
\usepackage{arxiv}
\usepackage[utf8]{inputenc} 
\usepackage{csquotes}  
\usepackage{amsmath}
\usepackage{array} 
\usepackage{longtable} 
\usepackage{colortab} 
\usepackage{colortbl}
\usepackage{arydshln}
\usepackage{import}
\usepackage[pdftex]{graphicx}
\usepackage{enumerate}
\usepackage{mcode}   
\usepackage{array} 
\usepackage{longtable} 
\usepackage{colortab} 
\usepackage{colortbl}
\usepackage{arydshln}
\usepackage{amsthm}
\usepackage{hyperref}
\usepackage{tikz}
\usepackage{bm}
\usepackage[colorinlistoftodos]{todonotes}
\usepackage{algorithm}
\usepackage{algpseudocode}
\usetikzlibrary{matrix,decorations, calc, positioning} 
\newtheorem{theorem}{Theorem}
\newtheorem{lemma}{Lemma}
\newtheorem*{remark}{Remark}
\usepackage[T1]{fontenc}    
\usepackage{lscape}
\usepackage{multirow}

\usepackage{amsmath,amsfonts,tikz-cd}

\begin{document}
\title{Finding the Exact Distribution of \\ (Peak) Age of Information for Queues of $PH/PH/1/1$ and $M/PH/1/2$ Type}
\author{
    Nail~Akar\\
	Electrical and Electronics Engineering Dept.\\
	Bilkent University\\
	Ankara, Turkey \\
	\texttt{akar@ee.bilkent.edu.tr} \\
	\And
    Ozancan~Doğan\\
    Electrical and Electronics Engineering Dept.\\
	Bilkent University\\
    Ankara, Turkey \\
    \texttt{ozancan@ee.bilkent.edu.tr} \\
    \And
    Eray~Ünsal~Atay\\
    Electrical and Electronics Engineering Dept.\\
	Bilkent University\\
    Ankara, Turkey \\
    \texttt{unsal.atay@ug.bilkent.edu.tr} \\ }

\maketitle
\begin{abstract}
 Bufferless and single-buffer queueing systems have recently been shown to be effective in coping with escalated Age of Information (AoI) figures arising in single-source status update systems with large buffers and FCFS scheduling.
In this paper, for the single-source scenario, we propose a numerical algorithm for obtaining the exact distributions of both the AoI and the peak AoI (PAoI) in (i) the bufferless $PH/PH/1/1/P(p)$ queue with probabilistic preemption with 
preemption probability $p$, $0 \leq p \leq 1$, and (ii) the single buffer $M/PH/1/2/R(r)$ queue with probabilistic replacement of the packet in the queue by the new arrival with replacement probability
$r$,  $0 \leq r \leq 1$.
The proposed exact models
are based on the well-established theory of Markov Fluid Queues (MFQ) and the numerical algorithms are matrix-analytical and they rely on numerically stable and efficient vector-matrix operations. Moreover, the obtained exact distributions are in matrix exponential form, making it amenable to calculate the tail probabilities and the associated moments straightforwardly. Firstly, we validate the accuracy of the proposed method with simulations, and for sume sub-cases, with existing closed-form results. 
We then comparatively study the AoI performance of the queueing systems of interest under varying traffic parameters. 
\end{abstract}
\section{Introduction}
\label{intro}
Consider the information update system in Fig.~\ref{fig:multisource} consisting of one information source equipped with a sensor, a server local to the source in the form of a queue, and a remote monitor (or destination).
The state of the information source is assumed to change in time which is detected by its sensor and the information source occasionally generates packets that contain sensed data along with a time stamp, to be immediately forwarded to the server. The server's role is to make decisions on when and which of these packets are to be sent towards the monitor which is responsible of collecting, monitoring, and further processing of these update messages. Packets are sent by the server to the monitor via a network which introduces random delays, i.e., service time of packets, and the monitor immediately sends back positive acknowledgments to the server. This single source model also studied in several other references including \cite{kaul_etal_infocom12},\cite{kosta_etal},\cite{kosta_etal_survey},\cite{inoue_etal_tit19} is the focus of the current paper.
The Age of Information (AoI) of a source is defined as the time elapsed since the generation of the last successfully received update packet at the monitor. 
The AoI concept was first introduced in \cite{kaul_etal_SMAN11} and  later elaborated in \cite{kaul_etal_infocom12},\cite{kaul_etal_ciss12} as a metric to quantify the freshness of knowledge about the status of a remote information source in a status update system. In these studies, the update system described above is viewed as an abstraction of real-world scenarios and applications in which data freshness is crucial. 
There has recently  been a surge of interest on AoI-related problems in various contexts including development of analytical models for AoI \cite{inoue_etal_tit19},\cite{costa_etal_TIT16},\cite{chen_huang_isit16}, AoI optimization methods \cite{huang_modiano},\cite{arafa_ulukus_asilomar17},\cite{sun_etal_tit17}, and AoI scheduling mechanisms \cite{hsu_etal_isit17},\cite{he_etal_TIT18}. The reference \cite{kosta_etal_survey} provides a relatively recent survey of the AoI concept and its applications.
The focus of the current paper is on the development of the queueing models of the particular single-source, single-server scenario, and the scenarios involving multiple information sources sharing a single server \cite{yates_kaul_ISIT12}, multiple servers \cite{bedewy_etal_isit16}, non-zero acknowledgment delays, and packet errors introduced by the network \cite{chen_huang_isit16}, are deliberately left outside the scope of
this paper.

\begin{figure}[tb]
	\centering
	\includegraphics[width=0.7\linewidth]{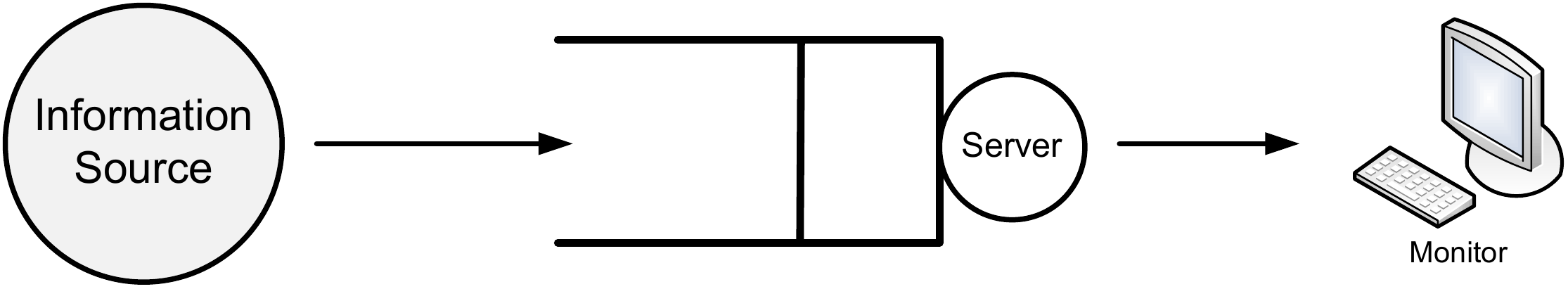}
	\caption{An information source sending status update messages through a queue to a remote monitor.} 
	\label{fig:multisource}
\end{figure}

The AoI processes of interest will first be introduced in the following very general setting not necessarily restricted to bufferless or single-buffer systems only.
Using certain buffer management and scheduling mechanisms, the server is to send (in their entirety) a fraction of these packets, named successful packets, either immediately upon arrival or after a queue wait, towards the monitor. In this general description, some packets, named unsuccessful packets, are to be dropped at the server after receiving partial service or without receiving any service at all.
Let $t_j$ ($t_j^{'})$ denote the arrival instant of the $j^{\text{th}}, j=1,2,\ldots,$ successful (unsuccessful) information packet arriving at the queue.
The interarrival time $\Lambda$ between information packet (successful or unsuccessful) arrival epochs is assumed to be independent and identically distributed with cumulative distribution function (cdf) $F_{\Lambda}(\cdot)$, 
probability density function (pdf) $f_{\Lambda}(\cdot)$, 
mean $1/\lambda$, and squared coefficient of variation (scov) $c_{\Lambda}^2 = \sigma_{\Lambda}^2 \lambda^2$ where $\sigma_{\Lambda}^2$ denotes the variance of $\Lambda$.
Let $\delta_j,j \geq 1$ denote the reception time at the monitor of the $j^{\text{th}}$ successful packet. 
Also, let $D_j,j \geq 1$ denote the system time of successfully received packet $j$, which is the sum of two terms: i) 
the queue wait time $W_j$, and ii) the service time $\Theta_j$, of the $j^{\text{th}}$ received packet. Clearly, $D_j= W_j + \Theta_j=\delta_j-t_j,j \geq 1$. In the current paper, all packet service times (denoted by $\Theta$) are assumed to be independent and identically distributed  with common cdf $F_{\Theta}(\cdot)$, pdf $f_{\Theta}(\cdot)$, mean $1/\mu$ and scov  $c_{\Theta}^2 = \sigma_{\Theta}^2 \mu^2$ where $\sigma_{\Theta}^2$ denotes the variance of $\Theta$.
The system load $\rho$ is defined as $\rho=\lambda/\mu$. 

Let $\Delta(t), t\geq 0,$ 
denote the continuous-time random process with left-continuous sample paths representing the AoI for this source at time $t$ with given $\Delta(0)$. After $t=0$, $\Delta(t)$ increases linearly in time with a unit slope until the first packet reception at $t=\delta_1$.
The right limit $\Delta(\delta_1^+)=\lim_{t \downarrow \delta_1} \Delta(t)$ is set to $D_1$ after which the process $\Delta(t)$ subsequently increases linearly in time with unit slope until the next reception and the pattern repeats forever. Let $\Phi_j = \Delta(\delta_j), j \geq 1,$ denote the discrete-time continuous-valued random process (called the PAoI process) associated with the AoI just at the epoch of packet receptions. 
Fig.~\ref{fig:SamplePath} shows a sample path of the random process $\Delta(t)$ for this general setting
with $\Delta(0)=0$ with five information packets arriving at a single buffer queue out of which the second packet is preempted by the third one. What is important to observe from Fig.~\ref{fig:SamplePath} is that each cycle of the AoI process consists of a linear curve that starts at value $D_j$ for some index $j$ and increases at a unit rate until the peak value $\Phi_{j+1}$ of that cycle and the sample path of the AoI process consists of an ordered concatenation of infinitely many cycles. The difficulty of building analytical models for these two AoI processes stems from the fact that the random variables $\Phi_j$ and $D_j$ are not independent and conventional queueing models do not cope well with jumps of this nature. For example, Markov Fluid Queues (MFQ) \cite{kulkarni_1997} that are instrumental to the current paper, can not produce sample paths given as in Fig.~\ref{fig:SamplePath} due to the above-mentioned difficulty. However, as to be shown in the sequel, MFQs can be devised to produce sample paths that contain sample cycles whose parts coincide with the sample cycles of the AoI process while losing the ordering relationship of the AoI cycles within a sample path. Moreover, these cycles contain sample values that coincide with the sample values of the PAoI process, again losing the ordering relationship among PAoI samples. In the current paper, we show that such ordering is not needed for the marginal distributions of the AoI and PAoI processes and it is this observation that leads us to use the well-established tool of MFQs for AoI and PAoI analysis. 
\begin{figure}[tb]
	\centering
	\includegraphics[width=0.8\linewidth]{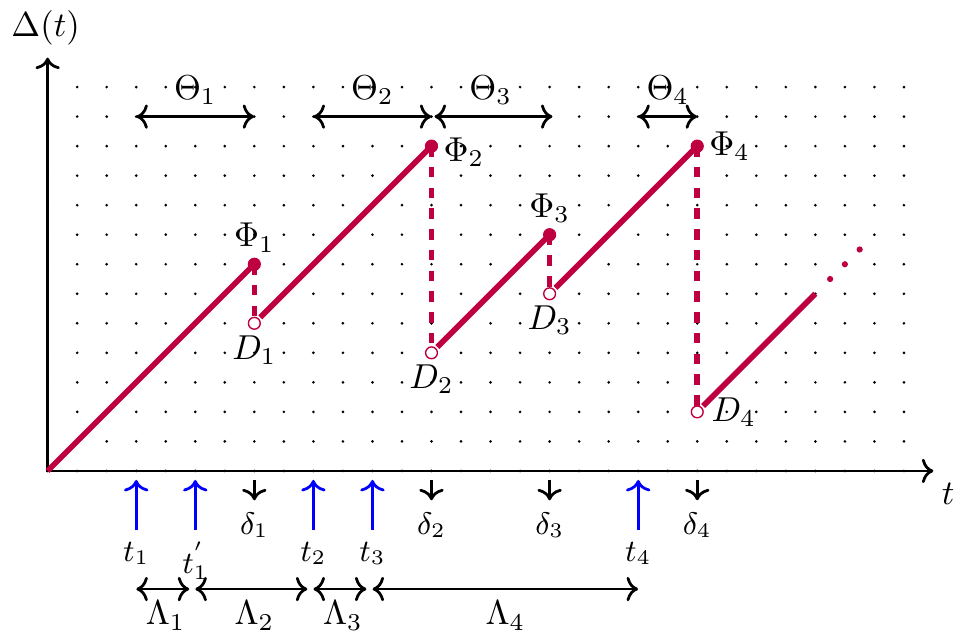}
	\caption{Illustration of the sample path for the AoI process $\Delta(t)$ and the PAoI process $\Phi_j$ for a single-buffer queueing system with five information packets arriving at the server from the source
		at epochs 3, 5, 9, 11, and 20, respectively out of which the second packet is preempted by the third packet, all other four packets being successful. 
		The successful packets 1 to 4 have a service time of 4, 4, 4, and 2, and a queue wait of 0, 0, 2, and 0, 
		respectively. $D_j$ is the system time, $\Theta_j$ is the service time, and $\delta_j$ is the reception time, of the $j^{\text{th}},j=1,\ldots,4$ successful packet, and $\Lambda_i$ is the interarrival time between the packet arrivals (successful or not) $i$ and $i+1$.
	} 
	\label{fig:SamplePath}
\end{figure}

Depending on the distribution of the interarrival time $\Lambda$, the service time $\Theta$, the queue capacity which represents the maximum number of packets that are allowed to be in the system (waiting room plus service), the scheduling discipline (First Come First Serve (FCFS), Preemptive Last Come First Serve (P-LCFS), Non-preemptive LCFS (NP-LCFS), etc.), and the buffer management scheme which refers to the specific decision of which packets are to be dropped at the server, there are several variations of the queueing system described above that have been studied in the recent literature. 
In most previous studies on AoI analysis, the focus has been on obtaining the mean AoI and mean PAoI values but there is also a need to obtain the distributions of the AoI and PAoI processes since their tail behaviors may also be important for the underlying status update system.
In \cite{kaul_etal_infocom12}, the mean AoI is obtained for the $M/M/1$, $M/D/1$, and $D/M/1$ queues with infinite buffer capacity and FCFS scheduling. The authors of \cite{inoue_etal_ISIT17} derive expressions
for the Laplace-Stieltjes transform of the stationary distributions
of the AoI and the PAoI in $M/GI/1$ and $GI/M/1$ queues. However, these infinite buffer queueing systems that are quite popular in other contexts tend to perform quite poorly in terms of AoI in moderate to high load regimes. Recognizing this fact, the reference \cite{costa_etal_TIT16} studies the AoI and PAoI distributions in the   $M/M/1/1$ and $M/M/1/2$ queues in which packets are dropped when the buffer is full at arrival epochs. 
The authors of \cite{costa_etal_TIT16} also introduce an alternative model, the so-called $M/M/1/2^{\ast}$ queue, for which the packet waiting in the queue is to be replaced by a new packet arrival, i.e., more formally named as $M/M/1/2/\mathit{NP\textnormal{-}LCFS}$ queue. This particular system is further studied in \cite{najm_nasser_isit16} as well as the $M/M/1/1/\mathit{P\textnormal{-}LCFS}$ preemptive queue where a new arrival preempts the packet in service and  the service time distribution is assumed to follow a gamma distribution and mean AoI and mean PAoI results are given. 
The reference \cite{inoue_etal_tit19} presents exact expressions for the stationary distributions of AoI and PAoI for a very wide class of single-source information update systems including bufferless and single buffer variations.
There are also several recent studies that attempt to obtain easily computable bounds for the crucial AoI-related metrics of interest.
As an example, \cite{soysal_ulukus_unpublished} derives exact expressions and upper bounds for the mean AoI for
$G/G/1/1$ and $G/G/1/1/\mathit{P\textnormal{-}LCFS}$ queues and they report that the upper bounds are in general close to exact average age expressions. Similarly, the authors of \cite{champati_etal_infocom19}  present a simple methodology for obtaining upper bounds for the AoI violation probability for both $GI/GI/1/1$ and $GI/GI/1/2^{\ast}$ systems, in addition to some exact closed-form expressions for some sub-cases.

A random variable (rv) corresponding to the time until absorption of a continuous-time Markov chain with one absorbing state is said to possess a phase-type (PH-type or PH) distribution; see \cite{neuts81} for properties of PH-type distributions and their applications to performance modeling.
Exponential, hyper-exponential, Erlang, and Coxian distributions, and their mixtures, are well-known examples of PH-type distributions.
One important property is that the set of PH distributions is dense in the field of all positive-valued distributions and PH distributions can therefore be used to approximate any positive-valued distribution \cite{ocinneide}. An Expectation Maximization (EM) algorithm for maximum likelihood estimation from sample data and density, for the purpose of approximation using PH distributions, is presented in \cite{asmussen_etal_SJS96} and very good fits to densities including Weibull, lognormal, etc. are obtained. The most well-known attribute of PH distributions is that problems arising in queueing systems that have an explicit solution assuming exponential distributions turn out to stay algorithmically tractable in case the exponential distribution is to be replaced with a PH distribution; see \cite{asmussen_etal_SJS96} and the references therein. In this way, using PH distributions allows one to algorithmically study the impact of higher order moments (such as the scov) of interarrival and service times on the performance of the queueing systems of interest. 

In this paper, we propose to use PH distributions to generalize the existing results on AoI for modeling interarrival and service times.
In particular, we propose a numerical algorithm for obtaining the exact distributions of both the AoI and PAoI processes as well as their moments for the following two systems:
\begin{enumerate}[(i)]
	\item The bufferless $PH/PH/1/1/P(p)$ queue with PH-type information packet arrivals and PH-type service times, and a packet in service is preempted by a new arrival with preemption probability $p$, $0 \leq p \leq 1$. When $p=0$, this model reduces to the ordinary  $PH/PH/1/1$ queue, whereas for $p=1$, the model is referred to as the LCFS preemptive $PH/PH/1/1/\mathit{P\textnormal{-}LCFS}$ queue, or shortly the $PH/PH/1/1^{\ast}$ queue. 
	\item The single buffer $M/PH/1/2/R(r)$ queue with Poisson arrivals and PH-type service times and with a waiting room of one packet only, and the packet in the waiting room is replaced with packet replacement probability $r$, $0 \leq r \leq 1$. 
	When $r=0$, this model reduces to the $M/PH/1/2$ FCFS queue, whereas for $r=1$, the model is referred
	to as the LCFS non-preemptive $M/PH/1/2/NP\textnormal{-}LCFS$ queue, or shortly the $M/PH/1/2^*$ queue named in \cite{costa_etal_TIT16}.
\end{enumerate}
For both queueing systems, we are interested in finding the following steady-state cdfs for the random processes $\Delta(t), t \geq 0,$ and $\Phi_j, j \geq 1$:
\begin{equation}	
F_{\Delta}(x) = \lim\limits_{t\to \infty }   \Pr \{ \Delta(t) \leq x \}, \ F_{\Phi}(x) = \lim\limits_{j \to \infty }   \Pr \{ \Phi_j \leq x \}, \ x \geq 0.
\end{equation}
Note that $F_{\Delta}(0)$ and $F_{\Phi}(0)$ must be zero since there can not be a probability mass at the origin for these two processes.
Also, let $f_{\Delta}(x)$ and $f_{\Phi}(x)$ for $x \geq 0$ denote the corresponding steady-state pdfs.
We have the following main contributions:
\begin{itemize}
	\item For the analytic modeling of AoI and PAoI, we propose to use a generalization of the well-established theory of MFQs whose mathematical foundation goes back to the works \cite{anick_mitra82},\cite{kosten.1984},\cite{kulkarni_1997}. Existing MFQ solvers that we use are matrix analytical and they rely on numerically stable and efficient vector-matrix operations. Moreover, the form of the obtained AoI and PAoI distributions is matrix exponential, making it amenable to calculate the tail probabilities, moments, etc. quite straightforwardly. This is in contrast with the need to invert Laplace transforms for finding age violation probabilities and the need for differentiating transforms for obtaining higher order moments as in \cite{inoue_etal_tit19}. 
	\item To the best of our knowledge, probabilistic preemption and replacement has not been analytically studied in the context of AoI distributions before. In this paper, we analytically model these probabilistic schemes and also demonstrate the merits of using probabilistic preemption in bufferless scenarios. Moreover, we have been able to provide a unifying algorithm for each of the bufferless and single-buffer queues through the probabilistic preemption and probabilistic replacement parameters.
\end{itemize}
The organization of the paper is as follows. In Section~2, preliminaries on PH distributions and MFQs are presented. Section~3 presents the proposed method for the two queueing systems of interest. In Section~4, we provide numerical examples to validate the proposed approach as well as the comparative assessment of the systems of interest under varying traffic parameters. Finally, we conclude in Section~5.

\section{Preliminaries}
\subsection{Notation}
Uppercase letters are used to denote real-valued matrices. Lowercase bold letters or symbols are used to denote real-valued vectors or scalars.
The $(i,j)^{\text{th}}$ th entry of $A$ is denoted by $A_{i,j}$ and the $j^{\text{th}}$ entry 
of a row or column vector $\bm{\alpha}$ is ${\alpha}_j$. 
The notations $\bm{0}_{k \times \ell} $, ${\bm I_m}$, and ${\bm 1_n}$ are used to denote the matrix of zeros of size $k \times l$, identity matrix of size $m$, and a column matrix of
ones of size $n$, respectively. The subscripts are dropped when the sizes are clear from the context. 
Let $A$ be an $n \times m$ matrix and $B$ a $p\times q$ matrix. The Kronecker product of the matrices $A$ and $B$ 
is denoted by  $A\otimes B$ which is of size $np \times mq$.
The notation $\textbf{diag}\{A,B\}$ denotes the block diagonal concatenation of the matrices $A$ and $B$ and is diagonal if the individual matrices $A$ and $B$ are diagonal.
A square matrix is said to be stable (anti-stable) if each of its eigenvalues has negative (non-negative) real parts. 
The notation $\left[\bm{\alpha},\bm{\beta}  \right]$ is used for the concatenation of the two row vectors $\bm{\alpha}$ and $\bm{\beta}$. 
$u(x)$ denotes the Heaviside step function, also known as the unit step function. $\delta(x)$ denotes the Dirac delta function, also known as the unit impulse function.
\subsection{Phase-type Distributions}
\label{ph}
To describe a
PH-type distribution, a Markov process is defined on the state-space
$\mathcal{S} = \{1,2,\ldots,m,m+1\}$ with one absorbing state $m+1$, initial probability vector $\left[ \bm{\sigma},\sigma_0\right] $, and an
infinitesimal generator of the form 
\( \left[ \begin{array}{{c;{2pt/2pt}c}}
S & \bm{\bm{\nu}} \\ \hdashline[2pt/2pt]
\bm{0} & 0 \end{array}
\right],
\)
where $\bm{\sigma}$ is a row vector of size $m$, $\sigma_0=1-\bm{\sigma} {\bm  1}$  is a scalar, the subgenerator
$S$ is $m \times m$, and $\bm{\bm{\nu}}$ is a column vector of size $m$ such that $\bm{\nu}=-S{\bm 1}$. The distribution of the time
till absorption into the absorbing state $m+1$, denoted by the random variable $X$, is called  PH-type characterized with the pair $(\bm{\sigma},S)$, i.e., $X \sim PH(\bm{\sigma},S)$. Very commonly, the probability mass at zero vanishes for PH-type distributions used for modeling interarrival times and service times, i.e., $\sigma_0=0$ \cite{neuts81}. 
Note that $\left[ \bm{\sigma}, \sigma_0 \right]$ is a probability vector 
and the diagonal elements of $S$ are strictly negative, its off-diagonal elements are non-negative and $S {\bm 1}\leq \bm{0}$ elementwise. 
The cdf and the pdf of $X \sim PH(\bm{\sigma},S)$, denoted by $F_X(x)$ and $f_X(x)$, respectively, are given as: 
\begin{equation}
F_X(x) =(1 -\bm{\sigma} e^{Sx} {\bm 1}) \:  u(x), \;
f_X(x) =-\bm{\sigma} e^{Sx} S {\bm 1} \:  u(x) + \sigma_0 \:  \delta(x).
\label{phdensity}
\end{equation}
A generalization of the PH-type distribution is the so-called Matrix Exponential (ME) distribution
\cite{AsmussenBladt97},\cite{he_aap07}. We say $X \sim ME(\bm{\sigma},S)$ with order $m$ when the pdf of the random variable $X$ is in the same form \eqref{phdensity} as in PH-type distributions, however, for the ME distribution, the parameters $\bm{\sigma}$ and $S$ do not necessarily have the same stochastic interpretation. 
The moments of $X \sim PH(\bm{\sigma},S)$ or $X \sim ME(\bm{\sigma},S)$ are easy to obtain:
\begin{equation}
\label{moments}
E[X^i]=i! \bm{\bm{\sigma}} (-S)^{-i} \bm{1}, \ i=1,2,\ldots.
\end{equation}
Despite the lack of stochastic interpretation of the row vector $\bm{\sigma}$ and the matrix $S$, ME distributions can be substituted in place of PH-type distributions in computational techniques and algorithms \cite{buchholz_telek_peva10}. 
We will also carry out a similar approach in this paper and use the more general ME distributions in place of PH-type distributions when necessary.

\begin{lemma}
	\label{lemma1}
	Let the rv $X$ have a pdf 
	$f_X(x)$ in the following matrix exponential form
	\begin{equation}
	f_X(x) = \bm{g} e^{Ax} \bm{h} \: u(x) + \sigma_0 \: \delta(x), \ E[X^i] = (-1)^{i+1} i! \bm{g} A^{-(i+1)} \bm{h}, \ i=1,2,\ldots,
	\label{statespace}
	\end{equation}
	with $A$ being of size $n$ and $\int_{-\infty}^{\infty} f_X(x) \ dx = -{\bm g} A^{-1} {\bm h} + \sigma_0=1$.
	Then, X is ME-distributed, i.e., there exists a vector $\bm{\sigma}$ and a matrix $S$ such that $X \sim ME(\bm{\sigma},S)$ with order $n$.
\end{lemma}
\begin{proof}
	When $h \neq 0$, there exists a nonsingular matrix $M$ such that $M{\bm 1} = {\bm v}$ where ${\bm v}=A^{-1}\bm{h}$.
	To show this, note that the eigenvalues of $A$ need to be in the open left half plane for $X$ to have a legitimate pdf. In particular, $A$ is non-singular and therefore ${\bm v}$ has at least one non-zero element.
	Let $k$ be the smallest index $k$ such that $v_k \neq 0$. We construct the matrix $M$ as follows. We set $M_{k,k}=v_k$ and for all $i \neq k$,
	if $v_i\neq0$, we set $M_{i,i}=v_i$, and if $v_i=0$, we set $M_{i,i}=1, M_{i,k}=-1$. All the remaining entries of $M$ are zero. Clearly, the matrix $M$ constructed as above is non-singular and  satisfies $M{\bm 1} ={\bm v}$. 
	Subsequently, one can easily show that $f_X(x)$ can be written in the form (\ref{phdensity}) with the choice of
	$\bm{\sigma}= \bm{g}M,\ S=M^{-1}AM$ and therefore $X\sim ME(\bm{\sigma},S)$.
\end{proof} 
\subsection{Markov Fluid Queues}
\label{mfq}
A Markov Fluid Queue (MFQ) is described by a joint Markovian process  
${\bm X(t)}= (X_f(t),X_m(t))$, $ t\geq 0$, where
$0\leq X_f(t) < \infty,X_m(t) \in {\mathcal S}= \{1,2,\ldots,n\}$,
$X_f(t)$ represents the continuous-valued fluid level in the buffer and the modulating phase process $X_m(t)$  is a Continuous Time Markov Chain (CTMC) with state space ${\mathcal S}$ and generator $Q$. 
In MFQs, the net rate of fluid change (or drift) is $r_i$ when the phase of the modulating process $X_m(t)$ is $i$. The drift matrix $R$ is the diagonal matrix of drifts: $ R=\textbf{diag}\{
r_1, r_2, \ldots, r_{n} \}$. When $X_f(t)=0$ and $X_m(t)=i$ with $r_i < 0$, $X_f(t)$ sticks to the boundary at zero.  The process ${\bm X(t)}$ is said to be characterized with the matrix pair $(Q,R)$, i.e., ${\bm X(t)} \sim MFQ(Q,R)$.
Stationary solution of infinite MFQs are studied in \cite{kulkarni_1997},\cite{anick_mitra82} by using the eigendecomposition of a  certain matrix.
The reference \cite{akar_sohraby_jap04} obtains the stationary solution of infinite and finite MFQs without having to find the eigenvectors, a problem which is known to be  ill-conditioned \cite{golub.vanloan.1996}.

The MFQ of interest to this paper is slightly different in the sense that  $X_m(t)$ behaves according to generator $Q$ as before when $X_f(t) > 0$ but it behaves according to another generator $\tilde{Q}$ when $X_f(t) =0$. This generalized MFQ, called a GMFQ throughout this paper, is characterized with the ordered  triple $(Q,\tilde{Q},R)$, i.e., ${\bm X(t)} \sim GMFQ(Q,\tilde{Q},R)$. The size of these matrices, $n$, gives the order of the GMFQ.
GMFQs turn out to be a special case of the more general multi-regime MFQs studied in \cite{kankaya.2008}. Therefore, their steady-state solutions are known and the underlying numerical methods are shown to be numerically stable and efficient.
We assume $r_i \neq 0,\ 1 \leq i \leq n$ which suffices for the AoI models developed in this paper.
We further assume without loss of generality that there are $b$ ($a$) states with positive (negative) drifts with $a+b=n$
and
$r_i > 0,\ i \leq b$ and $r_i < 0,\ i > b$, since otherwise states can always be reordered for this purpose.
For GMFQs, it is of interest the following steady-state joint pdf vector
\begin{equation}
\bm{f}(x)  =  \left[  f_1(x), f_2(x),\ldots,f_{n}(x) \right], \quad  f_i(x)  =  \lim\limits_{t\to \infty } \frac{d}{dx}  \Pr\{X_f(t)\leq x, X_m(t)=i\}, \; x >0 ,\label{density}
\end{equation}
and the steady-state probability mass accumulation (pma) vector at zero:
\begin{equation}
\bm{c} = \left[  c_1, c_2, \ldots, c_{n}  \right], \quad
c_i  =  \lim\limits_{t\to \infty } \Pr \{X_f(t)=0, X_m(t)=i\}. \label{accumulation}
\end{equation}
Alg.~1 describes the pseudo-code of obtaining the quantities of interest in \eqref{density} and \eqref{accumulation} on the basis of the numerical algorithm proposed in  \cite{kankaya.2008} for more general multi-regime MFQs in cases when  the steady-solution exists. 
\begin{algorithm}[t]
	\label{alg1}
	\caption{Steady-state Solution of $GFMQ(Q,\tilde{Q},R)$}
	\begin{algorithmic}[1]
		\Function{Steady-State}{$Q,\tilde{Q},R,n,a,b$}\Comment{$n$ is the size of these matrices}
		\State Step 1: Find an orthogonal matrix $P$ that puts the matrix $Q R^{-1}$ into the following form: \begin{equation}
		P^T Q R^{-1} P =\left[ \begin{array}{c;{2pt/2pt}c}
		F_{a \times a} & \ast  \\ \hdashline[2pt/2pt]
		\bm{0} & A_{b \times b}
		\end{array} \right], \quad P^T = \left[ \begin{array}{c} 
		\ast \\ \hdashline[2pt/2pt]
		H_{b \times n}
		\end{array} \right]  \label{step2}
		\end{equation}
		for an anti-stable matrix $F$ with an eigenvalue at the origin, stable matrix $A$, and $\ast$ denoting an arbitrary sub-matrix. The ordered real Schur form (available in Lapack, Matlab, and Octave software packages) is one alternative means of obtaining \eqref{step2}; see \cite{golub.vanloan.1996},\cite{akar_sohraby_peva09} and the references therein for the Schur form and its numerical stability.
		\State {Step 2:}  Solve for the $1 \times b$ vector $\bm{g}$ and $1 \times a$ vector $\bm{d}$ from the following linear matrix equation:
		\[
		\left[ \begin{array}{c;{2pt/2pt}c}
		\bm{g} & \bm{d}
		\end{array} \right]
		\left[ \begin{array}{c;{2pt/2pt}c}
		HR & -A^{-1}H{\bf 1}_n \\ \hdashline[2pt/2pt]
		-\tilde{Q}^{\ast} & {\bf 1}_a 
		\end{array} \right] = 
		\left[ \begin{array}{c;{2pt/2pt}c}
		{\bm 0}_{1 \times n} & 
		1
		\end{array} \right],
		\]
		with $\tilde{Q}^{\ast}$ denoting the matrix composed of the last $a$ rows of $\tilde{Q}$. 
		\State {Step 3:} Write $\bm{c}=\left[ {\bm 0},\bm{d} \right]$ and the steady-state joint pdf vector as \begin{equation}
		\bm{f}(x) = \bm{g} e^{Ax} H \ u(x) + \bm{c} \ \delta(x),  \ f_i(x)= \bm{g} e^{Ax} {\bm h}_i \ u(x)+ c_i \ \delta(x),  \label{meform}
		\end{equation}
		where ${\bm h}_i$ denotes the $i^{\text{th}}$ column of $H$.
		\EndFunction	
	\end{algorithmic}
\end{algorithm}
The computational complexity of the first two steps are $\mathcal{O}(n^3)$ in Alg.~1.
In an important sub-case that will be shown to arise in AoI queueing models later in this paper, a simple explicit way of complexity $\mathcal{O}(n^2)$ to find a matrix $P$ in \eqref{step2} of Alg.~1 is presented in the following lemma based on the Householder transformation given in \cite{golub.vanloan.1996}.
\begin{lemma}
	\label{householder}
	Consider the process ${\bm X(t)} \sim GMFQ(Q,\tilde{Q},R)$ with order $n$, $R = \textbf{diag} \{ \bm{I},-1\}$, and the last row of $Q$ is the zero matrix. Let $\bm{u_1}$ be a column vector of ones of size $n$ except for the last entry which is minus one. Also, let $\bm{u_2}$ be a column vector of zeros except for the first entry which is one. 
	Let $\bm{u}=\bm{u_1} - || \bm{u_1} ||_2 \bm{u_2}$. 
	Then, the symmetric orthogonal matrix $P$ defined by $P={\bm I} - \frac{2 \bm{uu^T}}{\bm{u^T u}}$ gives rise to the factorization \eqref{step2} with the scalar $F$ being zero. \end{lemma}

Other than the steady-state joint pdf and pma vectors, we are also interested in the following steady-state conditional density of the fluid level just before a visit from any state in a particular subset ${\mathcal{S}_0}\subset\mathcal{S}$ to one particular state $j \in \mathcal{S}\setminus{\mathcal{S}_0}$:
\begin{equation}
g_{j}^{{\mathcal{S}_0}}(x)= \frac{d}{dx} \lim_{\substack{t \to \infty \\  \Delta t \to 0}} \Pr\{X_f(t) \leq x \mid
X_m(t) \in {\mathcal{S}_0}, X_m(t+\Delta t) = j\}, \; x \geq 0.\label{gj}
\end{equation}
The following theorem provides a closed-form expression to obtain the conditional steady-state density $g_{j}^{{\mathcal{S}_0}}(x)$ from the steady-state joint density $f_i(x)$ of the MFQ  ${\bm X(t)}$.
\begin{theorem}
	\label{beginningtheorem}
	Let ${\bm X(t)} \sim GMFQ(Q,\tilde{Q},R)$ be of order $n$ and $\bm{f}(x)$ be its steady-state joint pdf vector given as in \eqref{meform}. Let ${\mathcal{S}_0}\subset\mathcal{S}$ be such that the fluid level does not have a probability mass at zero in any of the states in ${\mathcal{S}_0}$. Then, the steady-state conditional density $g_{j}^{{\mathcal{S}_0}}(x)$ is written in terms of the steady-state density vector $\bm{f} (x)$:
	\begin{equation}
	g_{j}^{{\mathcal{S}_0}}(x) 
	= \dfrac{ \sum\limits_{i\in{\mathcal{S}_0}} f_i(x)Q_{i,j} }{ \int\limits_{0}^{\infty} \sum\limits_{i\in{\mathcal{S}_0}} f_i(x^{\prime})Q_{i,j}dx^{\prime} } = 
	{\bm g^{\prime}} e^{Ax} {\bm h} \ u(x),
	\label{beginning_expression}  
	\end{equation}	
	where $\bm{h}=H \bm{\eta}$, $\eta_i = Q_{i,j}$ if $i \in {\mathcal{S}_0}$ and zero otherwise, and 
	$\bm{g^{\prime}} = {\bm g}/(-{\bm g} A^{-1} {\bm h})$.
\end{theorem}
\begin{proof}
	We first write
	\begin{align}
	&\lim_{\substack{t \to \infty \\  \Delta t \to 0}} \Pr\{X_f(t) \leq x \mid X_m(t) \in {\mathcal{S}_0}, X_m(t+\Delta t) = j\} = \nonumber\\
	& \lim_{\Delta t \to 0} \dfrac{\lim\limits_{t \to \infty} \Pr\left\{X_f(t) \leq x, X_m(t) \in {\mathcal{S}_0}, X_m(t+\Delta t) = j\right\}}{\lim\limits_{t \to \infty} \Pr\left\{X_m(t) \in {\mathcal{S}_0}, X_m(t+\Delta t) = j\right\}}.\label{yeni}
	\end{align}
	Note that as $\Delta t$ approaches $0$, the denominator of above can be written as:
	\begin{equation*}
	\begin{aligned}
	& 
	= \lim_{t \to \infty} \sum\limits_{i\in{\mathcal{S}_0}}\  \Pr \{ X_m(t+\Delta t)=j \mid X_m(t) = i \}  \cdot \Pr\{ X_m(t) = i \}, \\
	&=  \sum\limits_{i\in{\mathcal{S}_0}} Q_{i,j}  \Delta t \cdot \int\limits_{0}^{\infty} f_i(x^{\prime}) dx^{\prime} =
	\Delta t \int\limits_{0}^{\infty} \sum\limits_{i\in{\mathcal{S}_0}} f_i(x^{\prime})Q_{i,j}dx^{\prime}.
	\end{aligned}
	\end{equation*}
	On the other hand, as $\Delta t$ approaches $0$, the numerator of the expression \eqref{yeni} is written as:
	\begin{align}
	&
	=\lim_{t \to \infty} \sum\limits_{i\in{\mathcal{S}_0}} \Pr \{X_f(t)\leq x, X_m(t+\Delta t)=j \mid  X_m(t) = i \} \cdot \Pr \left\{ X_m(t) = i \right \},
	\nonumber\\
	&= \lim_{t \to \infty} \sum\limits_{i\in{\mathcal{S}_0}} \Pr \left\{X_f(t)\leq x \mid X_m(t) = i \right\} \cdot \Pr \left\{X_m(t+\Delta t)=j \mid X_m(t) = i \right\} \cdot \Pr \left\{ X_m(t) = i \right\}, \nonumber \\
	&= \lim_{t \to \infty} \sum\limits_{i\in{\mathcal{S}_0}} \dfrac{ \Pr \left\{X_f(t)\leq x ,\ X_m(t) = i \right\} }{ \Pr \left\{ X_m(t) = i \right\} } \cdot\Pr  \left\{X_m(t+\Delta t)=j \mid X_m(t) = i \right\} \cdot \Pr \left\{ X_m(t) = i \right\}, \nonumber \\
	&= \lim_{t \to \infty} \sum\limits_{i\in{\mathcal{S}_0}}\ \Pr \{X_f(t) \leq x , X_m(t) = i \}  \cdot\Pr \left\{X_m(t+\Delta t)=j \mid X_m(t) = i \right\}, \nonumber\\
	&= \Delta t \sum\limits_{i\in{\mathcal{S}_0}} Q_{i,j}\int\limits_{0}^{x}f_i(x^{\prime})dx^{\prime}. \nonumber
	\end{align}
	The first expression for $g_{j}^{{\mathcal{S}_0}}(x)$ in \eqref{beginning_expression} immediately follows. 
	The choice of ${\bm h}$ in the second expression  \eqref{beginning_expression} is that the numerator of the first expression in \eqref{beginning_expression} can be written in the form ${\bm g} e^{Ax} {\bm h} \ u(x)$. The choice of $\bm{g^{\prime}}$ is to ensure that the second expression is a legitimate density, i.e., the denominator of the first expression is $-{\bm g} A^{-1} {\bm h}$. 
\end{proof}
\section{Queueing Models for the AoI and PAoI Processes}
\subsection{Bufferless Queues}
In this subsection, we study the exact distributions of the AoI and PAoI processes in the $PH/PH/1/1/P(p)$ queue  
which is formally described as follows. We assume the interrarrival time $\Lambda \sim PH(\bm{\tau},T)$ with order $k$ with $\bm{\kappa}  = -T\bm{1}, \tau_0=1-\bm{\tau} {\bm 1}=0$, and the service time $\Theta\sim PH(\bm{\sigma},S)$ with order $\ell$ and $\bm{\nu}= -S {\bm 1}, \sigma_0=1-\bm{\sigma} {\bm 1}=0$. The information packet in service is to be preempted by a new information packet arrival with probability $p$, $0 \leq p \leq 1$. 

For the purpose of obtaining the exact distributions of the AoI and PAoI processes in this system, we construct a GMFQ process ${\bm X(t)}$ by which we have a single fluid level trajectory of infinitely many cycles where each cycle
begins with the arrival of an information packet into the system, proceeds until the reception 
of not the next but the second next successful information packet, and terminates with a final phase required for preparation for the next cycle. A sample path for the 
fluid level process $X_f(t)$ is depicted in Fig.~\ref{lcfs} which is constructed as follows. At the beginning of phase 1 (solid red
curve), an information packet, say packet $i$, arrives to an empty system and starts receiving service. During phase 1, this packet is in service and the fluid level $X_f(t)$ increases at a unit rate. 
When the packet $i$ is not preempted during its service time (for example in Cycle 2 in Fig.~\ref{lcfs}), a transition to phase 2 occurs  (dashed blue curve) which lasts until a new information packet arrival and in phase 2, the fluid level still increases at a unit rate. Upon a new arrival, a transition to
phase 3 occurs (dotted green curve) during which the fluid level continues to increase at a unit rate until the service of the new packet arrival is over. If this packet turns out to be preempted during
phase 3, the service time needs a reset. When phase 3 is over, the system transitions to phase 4 (dashed black curve) in which the fluid level is brought down to zero with
a drift of minus one after which the fluid level stays at level 0 for an exponentially distributed duration of time with an
arbitrary parameter (we use the unit parameter without loss of generality in this paper). When phase 4 is over, a cycle
is completed and a new cycle gets to start again comprising phases one to four.
On the other hand, if the first packet $i$ in a cycle is to be preempted by a new packet, a direct transition into phase 4 occurs without 
experiencing the phases 2 and 3. This situation happens twice in Cycle 1 of Fig.~\ref{lcfs}. 
Irrespective of the type of such cycles, our main observation is the following. 
The value of the process $X_f(t)$ at the beginning of phase 2 is the system time $D_j$ of a successful packet, say $j$. On the other hand, the value of $X_f(t)$ at the end of phase 3 (the reception epoch of packet $j+1$) is $\Phi_{j+1}$. 
Therefore, the concatenation of the blue-dashed and green-dotted curves (from the beginning of phase 2 to the end of phase 3, with the phases 1 and 4 of each cycle excluded) in Fig.~\ref{lcfs} 
is also a sample cycle of the AoI process in Fig.~\ref{fig:SamplePath}, and the fluid level just at the beginning of phase 4 in Fig.~\ref{lcfs} is also a sample value of the PAoI process in Fig.~\ref{fig:SamplePath}, and vice versa.
\begin{figure}[tb]
	\centering
	\includegraphics[scale=0.9]{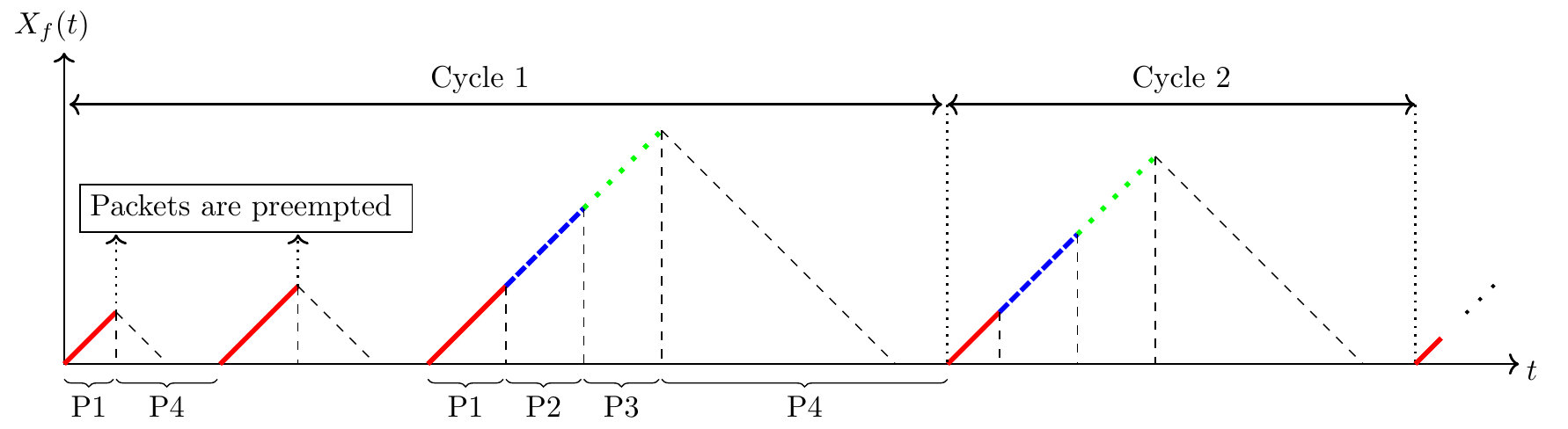}
	\caption{Sample path of the fluid level process $X_f(t)$ where $Pj$ denotes phase $j$, $j=1,\ldots,4$.}
	\label{lcfs}
\end{figure}
Moreover, the process $X_f(t)$ is the fluid level sub-process of a GMFQ process ${\bm X(t)}=(X_f(t),X_m(t))$ whose characterizing matrices we will now present. The state-space $\mathcal{S}$ of the modulating process $X_m(t)$ is written as 
$\mathcal{S}={\bigcup}_{n=1}^4\mathcal{S}_n$ where the set $\mathcal{S}_n$ refers to the set of states to be used for phase $n$ and is given as follows:
\begin{equation}
\mathcal{S}_m = \{ (i^{(m)},j^{(m)})\},\ i^{(m)}=1,\ldots,k, \ j^{(m)}=1,\ldots,\ell, \
\mathcal{S}_2 = \{ i^{(2)}\} ,\ i^{(2)}=1,\ldots,k, 
\label{S1-4}
\end{equation}
for $m=1,3$ and $\mathcal{S}_4 = \{ 0 \}$. For the set  $\mathcal{S}_m$, $i^{(m)}$ ($j^{(m)}$) keeps track of the state of the interarrival time (service time).  
With the lexicographical ordering of the states from $\mathcal{S}_1$ to $\mathcal{S}_4$,
${\bm X(t)} \sim GMFQ(Q,\tilde{Q},R)$ of order $2k \ell + k + 1$ where
\begin{equation}
Q =
\left[
\begin{array}{c;{2pt/2pt}c;{2pt/2pt}c;{2pt/2pt}c}
Q_{11} & {\bm I}_{k} \otimes \bm{\nu} & \bm{0} & p \bm{\kappa}  \otimes {\bm 1}_{\ell} \\ \hdashline[2pt/2pt]
\bm{0} & T & \bm{\kappa}  \otimes(\bm{\tau} \otimes \bm{\sigma}) & \bm{0} \\ \hdashline[2pt/2pt]
\bm{0} & \bm{0} & Q_{33} & {\bm 1}_{k} \otimes \bm{\nu} \\ \hdashline[2pt/2pt]
\bm{0} & \bm{0} & \bm{0}& {0} \\
\end{array}
\right],\ R = \textbf{diag}\{{\bm I}_{2k\ell+k},-1 \},
\label{Char1}
\end{equation}
where
\begin{equation}
Q_{11}  = {\bm I}_{k} \otimes S + T \otimes {\bm I}_{\ell} + (1-p)(\bm{\kappa}  \otimes\bm{\tau})\otimes {\bm I}_{\ell},  
\ Q_{33}  = Q_{11} +  p \bm{\kappa}  \otimes ({\bm 1}_{\ell} \otimes (\bm{\tau} \otimes \bm{\sigma})).
\end{equation}
This characterization is based on the verbal description of each of the four phases according to Fig.~\ref{lcfs}. To see this, let us assume that we are in a phase 1 state $(i^{(1)},j^{(1)})$.  A new arrival occurring with transition rate $\bm{\kappa} _{i^{(1)}}$ will lead us (i) to state 0 with probability $p$ which explains the term $p \bm{\kappa}  \otimes {\bm 1}_{\ell}$ (ii) to state 
$(v,j^{(1)})$ with probability $(1-p)\tau_v$ which explains the term $(1-p)(\bm{\kappa}  \otimes\bm{\tau})\otimes {\bf I}_{\ell}$.
When the service time is over which occurs with transition rate $\bm{\nu}_{j^{(1)}}$, we transition to a phase 2 state 
$i^{(2)}$ which is the same as $i^{(1)}$ which explains the term ${\bm I}_{k} \otimes \bm{\nu}$. Other state transitions without new arrivals or departures are reflected in $Q_{11}$. One can follow this procedure to justify the other blocks of the matrix $Q$. The expression for $R$ immediately follows the verbal description of the four phases in the sample paths.
$\tilde{Q}$ is a matrix of zeros except for the $1 \times k \ell$ block at the south-west corner of $Q$ which should be set to $\bm{\tau} \otimes \bm{\sigma}$ to start a new cycle and the scalar at the south-east corner should be set to minus one. 
To see this, we escape from state 0 to state $(i^{(1)},j^{(1)})$ with rate $\bm{\tau}_{i^{(1)}}\bm{\sigma}_{j^{(1)}}$ which explains the term $\bm{\tau} \otimes \bm{\sigma}$ and then the row sum of the last row is set to zero by this choice of the south-east corner. 
Due to the sample cycle arguments described above, the steady-state solution of ${\bm X(t)}$ enables one to obtain the exact distributions of the AoI and PAoI processes as stated in the next theorem. 
\begin{theorem}
	\label{bufferlesstheorem}
	Consider the $GMFQ(Q,\tilde{Q},R)$ with the characterizing matrices as defined in \eqref{Char1} with $f_s(x), s \in {\mathcal S}$ being the steady-state joint pdf for the $GMFQ(Q,\tilde{Q},R)$ and can be found using Alg.~1 and Lemma~\ref{householder} stemming from the structure of $Q$ and $R$. Then, the AoI and PAoI processes are ME-distributed and their pdfs, denoted by $f_{\Delta}(x)$ and $f_{\Phi}(x)$, respectively, for the $PH/PH/1/1/P(p)$ queue, are given in terms of $f_s(x)$ as follows:
	\begin{equation}
	f_{\Delta}(x)=\frac{\sum\limits_{s \in {\mathcal S}_2 \bigcup{\mathcal S}_3} f_s(x)}{\int\limits_{0}^{\infty}\sum\limits_{s \in {\mathcal S}_2 \bigcup{\mathcal S}_3} f_s(x^{\prime}) \: dx^{\prime} }, \quad
	f_{\Phi}(x)=\dfrac{\sum\limits_{i=1}^k \sum\limits_{j=1}^{\ell} f_{(i^{(3)},j^{(3)})}(x) \bm{\nu}_j}
	{\int\limits_{0}^{\infty} \sum\limits_{i=1}^k \sum\limits_{j=1}^{\ell} f_{(i^{(3)},j^{(3)})}(x^{\prime}) \bm{\nu}_j \: dx^{\prime}}, \ x \geq 0.
	\label{expression1}
	\end{equation}
\end{theorem}
\begin{proof}
	The proof follows sample path arguments. The expression for $f_{\Delta}(x)$ stems from the observation that the restriction of the $GMFQ(Q,\tilde{Q},R)$ to the states in $\mathcal{S}_2$
	and $\mathcal{S}_3$, as indicated by the blue-dashed and green-dotted parts of the curve in Fig.~\ref{lcfs}, comprises the cycles of the actual AoI process. On the other hand, $f_{\Phi}(x)$ amounts to the pdf of the fluid level just at the epoch of transitions to state 0 from the states in $\mathcal{S}_3$ and is therefore given by the expression in \eqref{expression1}  using Theorem~\ref{beginningtheorem}. The ME-distributed nature of the stationary AoI and PAoI is an immediate result of Step~3 of Alg.~1 as well as Theorem~1.
\end{proof}
\subsection{Single Buffer Queues}
In this subsection, we study the $M/PH/1/2/R(r)$ queue model without preemption, and with a waiting room for one single
packet only. An information packet joins the queue upon arrival if the buffer is empty but the packet in the buffer is to be replaced  with probability $r$, $0 \leq r \leq 1$, by the new incoming information packet arrival.  
We assume the interrarrival time is exponentially distributed with parameter $\lambda$ and the service time $\Theta\sim PH(\bm{\sigma},S)$ with order $\ell$ and $\bm{\nu} = -S{\bm 1},\sigma_0=1-\bm{\sigma} {\bm 1}=0$.
\begin{figure}[tb]
	\centering
	\includegraphics[width=0.8\linewidth]{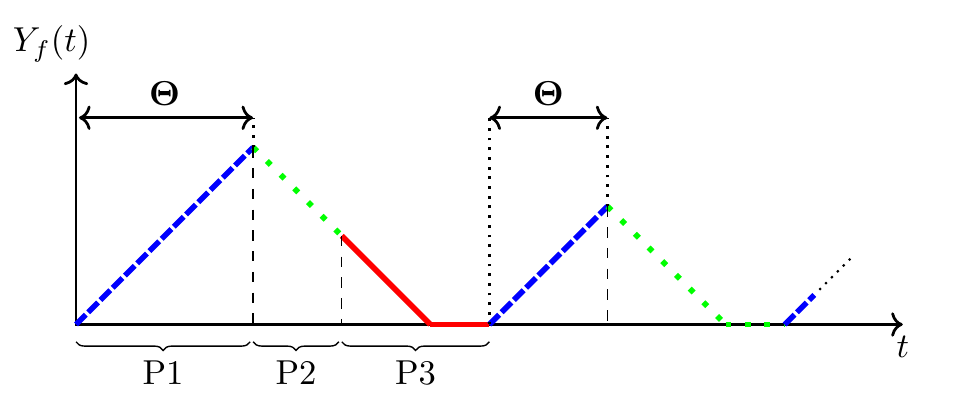}
	\caption{Sample path of the fluid level process $Y_f(t)$ corresponding to the residual service time.  $Pj$ denotes phase $j$, $j=1,\ldots,3$.}
	\label{unfinished}
\end{figure}
In order to derive the pdf of the queue wait time $W$ for successful packets, denoted by $f_W(x), x \geq 0$, we introduce a new fluid level process $Y_f(t)$ a sample path of which is given in Fig.~\ref{unfinished} that is related to the residual service time. For this purpose, let us assume a packet arrival at $t=0$ to an empty queue. Upon the arrival, the fluid process starts to increase at a unit rate for a duration that equals to the service time. This behavior is said to occur in phase 1 (shown by the dashed blue curve). Once the service time is over, a transition to phase 2 (shown by the dotted green curve) occurs in which the fluid level is allowed to decrease with a rate of minus one. If a new arrival occurs in phase 2 when $Y_f(t) >0$, we transition to phase 3 (shown by the solid red curve) during which $Y_f(t)$ is reduced at a unit rate irrespective of new arrivals. If we hit zero without transitioning to phase 3, with a new arrival, a transition from phase 2 to phase 1 occurs. When we hit zero at phase 3, we stay at level 0 for an exponentially distributed duration of time with unit parameter after which we transition to phase 1. If we exclude phase 1 and part of phase 3 corresponding to $Y_f(t)=0$ in this sample path, what remains is a sample path for the residual service time, i.e., the time needed to drain the packet in service, for both queues. The residual service time is zero if there is no packet in service.

Actually, $Y_f(t)$ is the fluid level subprocess of a GMFQ process ${\bm Y(t)}=(Y_f(t),Y_m(t))$ where
$Y_m(t) \in \mathcal{S}=\bigcup_{n=1}^3\mathcal{S}_n$ where 
$ \mathcal{S}_1 = \{  1^{(1)},\ldots,\ell^{(1)}\}$,
$\mathcal{S}_2 = \{ 0\}$,  and
$\mathcal{S}_3 = \{ 1 \}$.
With the enumeration of the states from $\mathcal{S}_1$ to $\mathcal{S}_3$ as before,
the proposed GMFQ is characterized with the triple $(Q,\tilde{Q},R)$ of order $ \ell +2$ where
\begin{equation}
Q =
\left[
\begin{array}{c;{2pt/2pt}c;{2pt/2pt}c}
S & \bm{\nu} & \bm{0}\\ \hdashline[2pt/2pt]
\bm{0}& -\lambda & \lambda \\ \hdashline[2pt/2pt]
\bm{0}& 0 & 0
\end{array}
\right], \quad
\tilde{Q} =   
\left[
\begin{array}{c;{2pt/2pt}c;{2pt/2pt}c}
\bm{0}& \bm{0} & \bm{0} \\ \hdashline[2pt/2pt]
\lambda\bm{\sigma} & -\lambda&{0} \\ \hdashline[2pt/2pt]
\bm{\sigma} & {0}& -1
\end{array}
\right] ,\quad R = \textbf{diag}\{\bm{I}_{l},-\bm{I}_{2} \}.
\end{equation}
Let $f_s(x), s \in {\mathcal S}$ denote the steady-state joint density for the $GMFQ(Q,\tilde{Q},R)$ obtained using the method described in subsection~\ref{mfq}.
Based on Step~3 of Alg.~1, we have the following matrix-exponential form for $f_i(x)$:
\begin{equation}
f_i(x)= \bm{g} e^{Ax} \bm{h}_ i  u(x) + c_i  \delta(x),\ i=0,1,  \label{meform5}
\end{equation}
for a row vector $\bm{g}$ of size $\ell$, a square matrix $A$ of size $\ell$, a column vector $\bm{h}_ i$ of size $\ell$, and two scalars $c_0$ and $c_1$ corresponding to the probability masses at states 0 and 1, respectively. The following theorem gives an expression for $f_W(x)$ for the two queueing systems of interest based on the expression \eqref{meform5}.
\begin{theorem}
	The pdf of the queue wait time $f_W(x)$ for the $M/PH/1/2/R(r)$ queue is given in terms of the matrix parameters of \eqref{meform5} as follows:
	\begin{equation}
	\label{Wexpression}
	f_{W}(x) = \eta_1 \left( \bm{g} e^{(A-r\lambda \bm{I}_{\ell})x}(\bm{h}_ 0+r \bm{h}_ 1) \ u(x) + c_0  \ \delta(x)  \right), 
	\end{equation}
	where $\eta_1=1/\left(-\bm{g}(A-r \lambda {\bm I}_{\ell})^{-1}(\bm{h}_ 0 + r \bm{h}_ 1)+ c_0 \right)$.
\end{theorem}
\begin{proof}
	Let $\tilde{f}_i(x), \ i=0,1$ denote the joint pdf of the residual service time and  there are $i$ information packets in the queue.
	We censor out the states in $\mathcal{S}_1$ and the probability mass at zero for state 1 from the original GMFQ to obtain the following expression for $\tilde{f}_i(x)$:
	\begin{equation*}
	\tilde{f}_0(x) = {\eta} \left( \bm{g} e^{Ax} \bm{h}_ 0   \ u(x) + c_0 \ \delta(x) \right), \quad \tilde{f}_1(x)= {\eta} \bm{g} e^{Ax} \bm{h}_ 1 \  u(x), \;x \geq 0,
	\end{equation*}
	where ${\eta} = 1/\left( -\bm{g}A^{-1}(\bm{h}_ 0 + \bm{h}_ 1) + c_0 \right)$. 
	For the $M/PH/1/2/R(r)$ queue, a successful 
	packet arrival occurs when the following conditions hold: (i) there are no information
	packets in the queue or existing packet is replaced with this new arrival, and (ii) this packet arriving when
	the residual service time equals $x \geq 0$, will not be replaced with
	another packet which occurs with probability $e^{-r \lambda x}$. 
	Therefore, for the $M/PH/1/2/R(r)$ queue, \[
	f_W(x)=  \frac{e^{-r\lambda x} }{\pi_s} \left( \tilde{f}_0(x) + r \tilde{f}_1(x)
	\right), \ x\geq 0,
	\] 
	where $\pi_s = \int_{-\infty}^{\infty} (\tilde{f}_0(x) + r \tilde{f}_1(x)) e^{-r \lambda x} dx$ is the success probability which consequently yields the expression \eqref{Wexpression} completing the proof.
\end{proof}
\begin{remark}
	The scalar term $e^{-r\lambda x}$ in the above proof arising when arrivals are Poisson, commutes with the matrix $e^{Ax}$ giving rise to the simple expression in Theorem~3. This is the main reason that hindered us from using more general PH-type interarrivals for single-buffer systems in which case such commutativity would not hold.
\end{remark}
The pdf given in the identity \eqref{Wexpression} is in matrix exponential form. Therefore, $W$ can be shown to be ME-distributed and $W \sim ME(\bm{\beta} ^{\ast}, B^{\ast})$ of order $\ell$ for some vector $\bm{\beta} ^{\ast}$ and matrix $B^{\ast}$ using the construction method described in Lemma 1. 
Let us first assume $W \sim PH(\bm{\beta} ,B)$ with $\bm{\psi} = -B{\bm 1}, {\beta}_0=1-\bm{\beta}  {\bm 1}$ with order $\ell$. 
Note that computational techniques and algorithms using PH-type distributions can be extended to ME distributions as well, as shown in various studies; see \cite{buchholz_telek_peva10} and the references therein. Therefore,  
after developing a numerical algorithm for obtaining $f_{\Delta}(x)$ and $f_{\Phi}(x)$ under the assumption $W \sim PH(\bm{\beta} ,B)$, we will then substitute $(\bm{\beta} ^{\ast}, B^{\ast})$ in place of $(\bm{\beta} ,B)$ 
in the numerical algorithm when $W \sim ME(\bm{\beta} ^{\ast}, B^{\ast})$.

A sample path of the fluid level process $Z_f(t)$ we propose to use for the 
$M/PH/1/2/R(r)$ queue is depicted in Fig.~\ref{mph12}. 
Similar to the $PH/PH/1/1/P(p)$ case, we have a single fluid level trajectory of infinitely many cycles where each cycle begins with the arrival of a successful information packet into the system and ends with the reception of the next successful information packet at the destination.
Let us assume that at $t=0$, a successful information packet arrival has just occured. In phase 1 (shown by the solid blue curve), the fluid level $Z_f(t)$ increases at a unit rate until its queue wait time $W$ is over. Once phase 1 is over, phase 2 starts (shown by the dashed red curve) which lasts until either a new packet arrival occurs or until the service time of the current packet is over. In the former case, we transition into phase 3 (shown by the solid black curve) during which we wait for the end of the service time before transitioning to phase 4. In the latter case, we transition from phase 2 to directly phase 4 (shown by the dotted green curve) in which case it lasts until a new packet arrival. 
In phase 5 (shown by the solid brown curve), the fluid level continues to increase at a unit rate until the service time of the new packet arrival is over. When phase 5 terminates, the system transitions to phase 6 in which the fluid level is brought down to zero with a drift of minus one after which the fluid level stays at level 0 for an exponentially distributed duration of time with unit parameter.
Similar to the sample path observations made for the bufferless system,
the parts of the curves from the beginning of phase 4 until the end of phase 5 in each cycle in Fig.~\ref{mph12} is also a sample cycle of the AoI process of Fig.~\ref{fig:SamplePath} for the single buffer queueing system, and the fluid level just at the beginning of phase 6 in Fig.~\ref{mph12} is also a sample value of the PAoI process of Fig.~\ref{fig:SamplePath}, and vice versa. Actually, $Z_f(t)$ is the fluid level sub-process of a GFMQ ${\bm Z(t) = (Z_f(t),Z_m(t))}$ where
$Z_m(t) \in \mathcal{S}=\bigcup_{n=1}^6\mathcal{S}_n$ where the set $\mathcal{S}_n$ refers to the collection of states to be used for phase $n$. In particular,
$\mathcal{S}_m=\{ i^{(m)} \},\ i^{(m)}=1,\ldots,\ell, \ m=1,2,3,5$,  
$\mathcal{S}_4 =\{ 1 \}$, and  
$\mathcal{S}_6 = \{ 0 \}$. 
\begin{figure}[tb]
	\centering
	\includegraphics[width=0.9\linewidth]{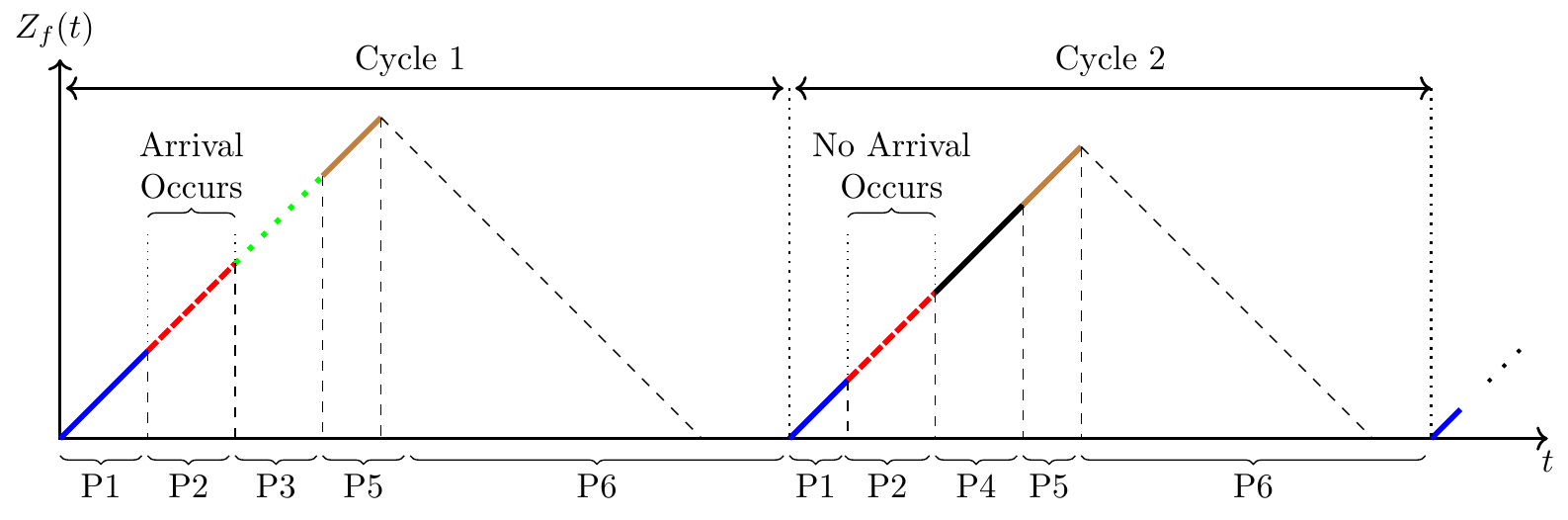}
	\caption{Sample path of the fluid level process $Z_f(t)$ where $Pj$ denotes phase $j$, $j=1,\ldots,6$.}
	\label{mph12}
\end{figure}
With the enumeration of the states from $\mathcal{S}_1$ to $\mathcal{S}_6$ as before,
the proposed GMFQ is characterized with the triple $(Q,\tilde{Q},R)$ of order $ 4 \ell +2$ where
\begin{equation}
Q =
\left[
\begin{array}{c;{2pt/2pt}c;{2pt/2pt}c;{2pt/2pt}c;{2pt/2pt}c;{2pt/2pt}c}
B & \bm{\psi} \otimes \bm{\sigma} & \bm{0} &\bm{0} &\bm{0} & \bm{0}\\ \hdashline[2pt/2pt] 
\bm{0}& S-\lambda {\bm I_{\ell}} & \lambda {\bm I_{\ell}} & \bm{\nu}&\bm{0}&\bm{0}\\ \hdashline[2pt/2pt]
\bm{0}&\bm{0}& S &\bm{0}& \bm{\nu} \otimes \bm{\sigma} &\bm{0}\\ \hdashline[2pt/2pt]
\bm{0}& \bm{0}&\bm{0}&-\lambda & \lambda \bm{\sigma} & 0\\ \hdashline[2pt/2pt]
\bm{0}& \bm{0}&\bm{0}& \bm{0}& S & \bm{\nu} \\ \hdashline[2pt/2pt]
\bm{0}&\bm{0}& \bm{0}&0& \bm{0}& 0
\end{array}
\right], \ R =\textbf{diag}\{\bm{I}_{4 \ell+1},-1 \},
\label{Qnail}
\end{equation}
and $\tilde{Q}$ is a matrix of zeros except for the $1 \times 2\ell$ block at the south-west corner which should be set to $\left[ \bm{\beta} , {\beta}_0 \bm{\sigma} \right]$ and the scalar at the south-east corner should be $-1$.
Let $f_s(x), s \in {\mathcal S}$ denote the steady-state joint density for the $GMFQ(Q,\tilde{Q},R)$. The expressions for the pdfs of the AoI and PAoI processes 
are given in the following theorem, a proof of which is similar to that of Theorem~2 and is omitted.
\begin{theorem}
	\label{buffertheorem}
	Consider the process $\bm{Z(t)} \sim GMFQ(Q,\tilde{Q},R)$ with the characterizing matrices as defined in \eqref{Qnail} with $f_s(x), s \in {\mathcal S}$ denoting 
	the associated steady-state joint pdf. The pdfs of the steady-state AoI and PAoI, 
	denoted by $f_{\Delta}(x)$ and $f_{\Phi}(x)$, respectively, for the $M/PH/1/2/R(r)$ queue, are in matrix exponential form and are given in terms of $f_s(x)$ as follows:
	\begin{equation}
	f_{\Delta}(x)=\frac{\sum\limits_{s \in {\mathcal S}_4 \bigcup {\mathcal S}_5} f_s(x)}
	{\int\limits_{0}^{\infty}\sum\limits_{{\mathcal S}_4 \bigcup {\mathcal S}_5}  f_s(x^{\prime}) dx^{\prime}}, \quad
	f_{\Phi}(x)=\dfrac{ \sum\limits_{j=1}^{\ell} f_{j^{(5)}}(x) \bm{\nu}_j}
	{\int\limits_{0}^{\infty}  \sum\limits_{j=1}^{\ell} f_{j^{(5)}}(x^{\prime}) \bm{\nu}_j dx^{\prime} }, \ x\geq 0.
	\label{expressionnail}
	\end{equation}
\end{theorem}
\section{Numerical Examples}
In order to make the numerical experiments reproducible, we made our Matlab implementation available at MATLAB File Exchange with the title
\href{https://www.mathworks.com/matlabcentral/fileexchange/74923-markov-fluid-queue-solver-for-age-of-information-dist}{``Markov Fluid Queue Solver for Age of Information Distribution"}.

\subsection{Validation with Simulations}
In the first example, 
we assume Poisson arrivals with rate $\lambda$ and fix $p=0$ or $p=1$ for the bufferless system and $r=0$ or $r=1$ for the single-buffer system, and then obtain the steady-state distributions of the AoI and the PAoI of the underlying four queueing systems, namely $M/PH/1/1$, $M/PH/1/1^{\ast}$, $M/PH/1/2$, 
and $M/PH/1/2^{\ast}$ queues, and compare with simulations.
For the service times, we use a PH-type distribution with mean $\rho/\lambda$ for a given system load $\rho$ and scov of the service time is to be fixed to a given value $c_{\Theta}^2$ according the following procedure: 
For $c_{\Theta}^2=1/j \leq 1$ for a positive integer $j$, $E(\mu^{-1},j)$ distribution is used which refers to an Erlang distribution with mean $\mu^{-1}$ and with order $j$. If $j$ is not an integer, then we resort to a mixture of two appropriate Erlang distributions \cite{tijms_book03}. When $c_{\Theta}^2 > 1$, then we propose to use a hyper-exponential distribution with balanced means to fit the first two moments \cite{tijms_book03}. 
Fig.~\ref{fig:fig6} depicts 
the cdfs of the AoI and PAoI processes obtained with the proposed analytical model as well as simulations, for four different values of the parameter pair $(\rho,c_{\Theta}^2)$, and for the four queueing models and perfect match with the simulation results
is obtained in all cases. 
\begin{figure*}[t]
	\centering
	\includegraphics[width= \linewidth]{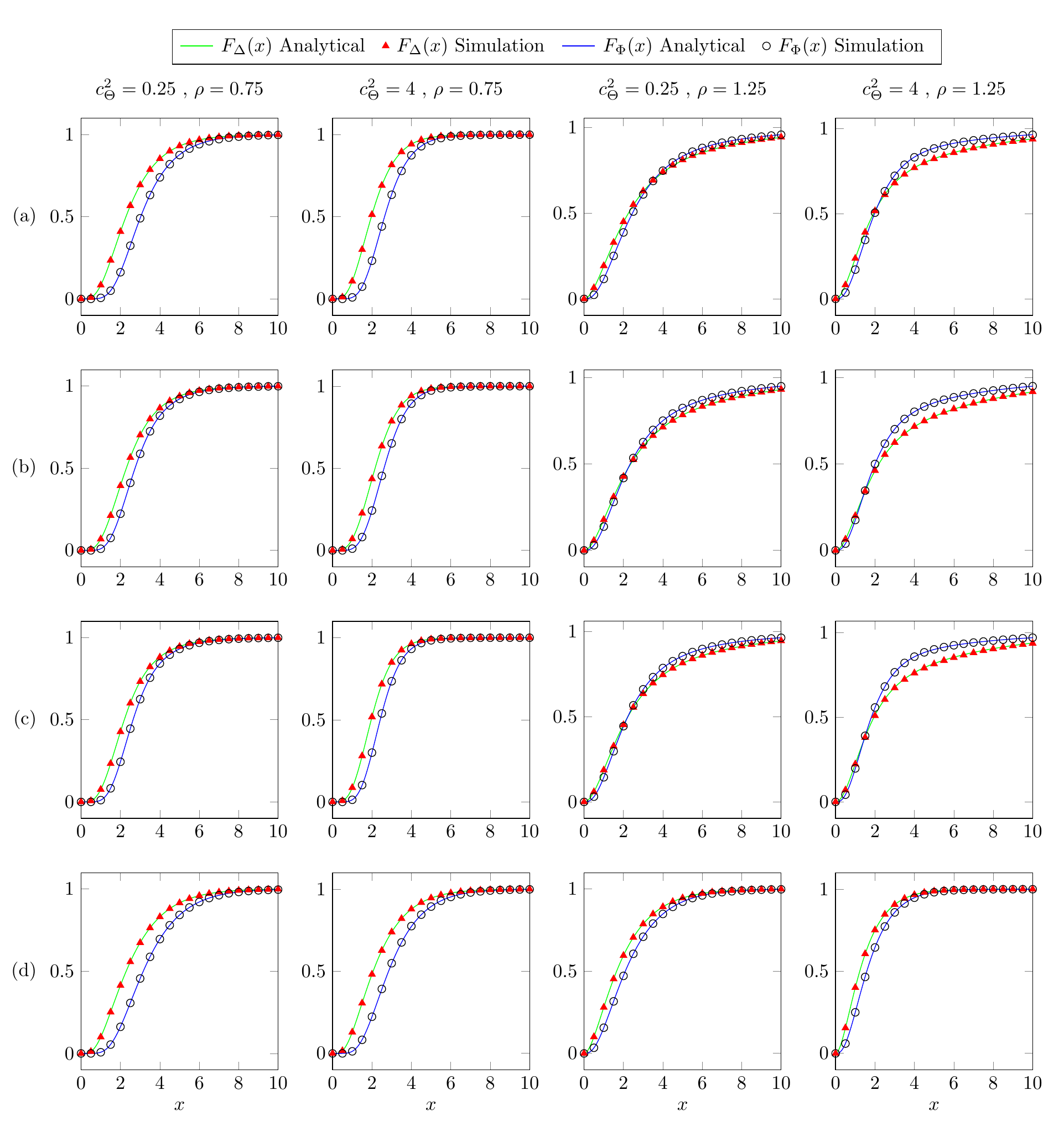}
	\caption{The cdfs of the AoI and PAoI processes obtained with the 
		proposed method and simulations for the 
		scenarios $\rho=0.75,1.25$, $c_{\Theta}^2=0.25,4$ for
		(a) $M/PH/1/1$ (b) $M/PH/1/1^{\ast}$ (c) $M/PH/1/2$ 
		(d) $M/PH/1/2^{\ast}$ queueing models.}
	\label{fig:fig6}
\end{figure*} \par
As a second example, we allow PH-type information packet arrivals for the bufferless case and the same construction of the arrival process is used as in the previous example to match $\lambda$ and $c_{\Lambda}^2$.
The cdfs of the AoI and PAoI processes are plotted in Fig.~\ref{fig:fig7} using 
the proposed analytical model and simulations for $\rho=0.75,1.25$, $c_{\Lambda}^2=0.25,4$ 
while fixing $c_{\Theta}^2=0.2$ for the two queueing models 
$PH/PH/1/1$ and $PH/PH/1/1^{\ast}$.
The results obtained by the 
analytical method perfectly match those obtained by simulations.
\begin{figure*}[t]
	\centering
	\includegraphics[width=\linewidth]{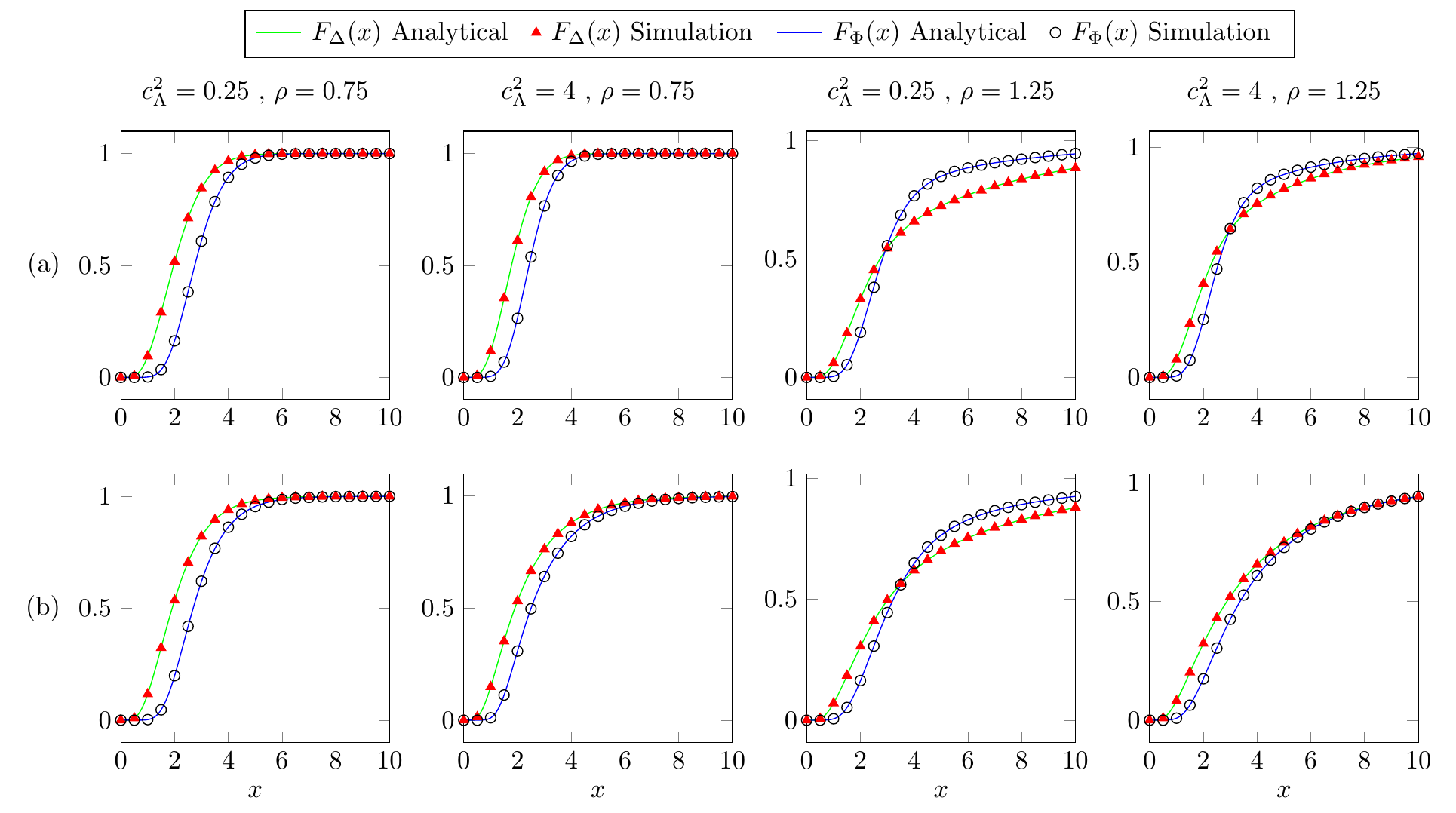}
	\caption{The cdfs of the AoI and PAoI processes when 
		$c_{\Theta}^2=0.2$, $\rho=0.75,1.25$, $c_{\Lambda}^2=0.25,4$ for the 
		(a) $PH/PH/1/1$ (b) $PH/PH/1/1^{\ast}$ queueing models.}
	\label{fig:fig7}
\end{figure*}

As a final example of this subsection, the mean AoI, $E[\Delta]$, is plotted in Fig.~\ref{fig:fig8} (obtained with both the proposed method and simulations which perfectly match) with respect to the packet preemption probability $p$ 
for the $M/PH/1/1/P(p)$ queue and with respect to the packet replacement probability $r$ for the $M/PH/1/2/R(r)$ queue, each queue with varying values of $\rho$  and $c_{\Theta}^2$. From the examples we studied, we observe that there is an optimal choice of the preemption probability (not necessarily 0 or 1) which minimizes $E[\Delta]$ and this optimal value $p^{\ast}$ depends on the pair $(\rho,c_{\Theta}^2)$. As a general remark, the complete preemption policy $p=1$ turns out to perform relatively poorly for larger values of $\rho$ and smaller values of $c_{\Theta}^2$. On the other hand, as expected, the complete replacement policy $r=1$ gives rise to the minimum $E[\Delta]$ for all the examples we studied.
\begin{figure*}[t]
	\centering
	\includegraphics[width=\linewidth]{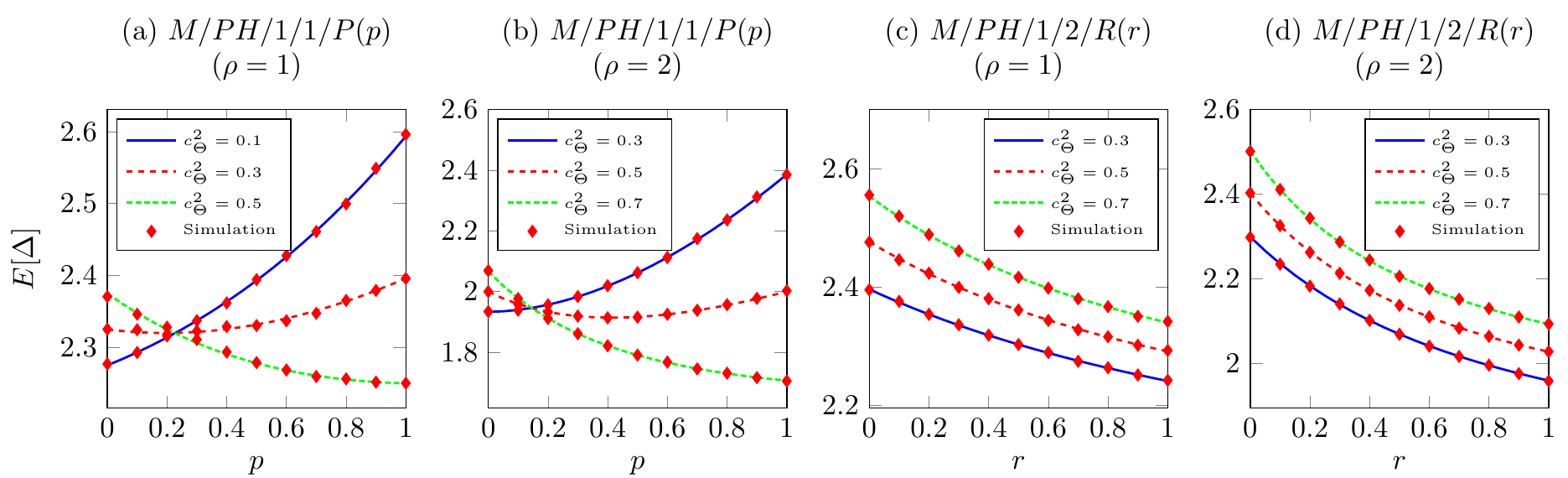}
	\caption{The mean AoI, as a function of the packet preemption probability $p$ for the $M/PH/1/1/P(p)$ queue for three values of $c_\Theta^2$: (a) $\rho=1$ (b)  $\rho=2$,  and as a function of the packet replacement probability
		for the $M/PH/1/2/R(r)$ queue again for three values of $c_\Theta^2$: (c) $\rho=1$ (d) $\rho=2$.}
	\label{fig:fig8}
\end{figure*}
\subsection{Validation with Existing Results}
In addition to comparisons with simulations, we validate the accuracy of our proposed method with a number of closed-form expressions existing in the literature. As an example, closed-form expressions for $E[\Delta]$ and $E[\Delta^2]$ are available in \cite{inoue_etal_tit19} for the $M/GI/1/1^{\ast}$ and $GI/M/1/1^{\ast}$ queues as well as an expression for $E[\Delta]$ for the $M/GI/1/2^{\ast}$ queue.
In Table~\ref{table2}, we tabulate the results up to four fractional digits, obtained with our proposed method and by the closed-form expressions given in \cite{inoue_etal_tit19} for various cases. The results perfectly match in all digits and for all the cases we tried, demonstrating the accuracy of the proposed method.
\begin{table}[t]
	\centering
	\caption{The performance metrics $E[\Delta]$ and $E[\Delta^2]$ obtained with the closed-form expressions in Ref.~\cite{inoue_etal_tit19} and the proposed method for various queueing models and their parameters.}
	\begin{tabular}{cccccc}
		\hline 
		& & \multicolumn{2}{c}{$E[\Delta]$}& \multicolumn{2}{c}{$E[\Delta^2]$}  \\ \hline
		\multicolumn{2}{c}{Queueing Model}& Ref.~\cite{inoue_etal_tit19} & Proposed  & Ref.~\cite{inoue_etal_tit19} & Proposed  \\ \hline
		$M/PH/1/1^{\ast}$ & $\lambda=0.5, \ \Theta \sim E(1,2)$ & 3.1250 &3.1250 & 14.5312 &14.5312 \\ \cline{2-6}
		& $\lambda=0.5, \ \Theta \sim E(1,4)$ & 3.2036 &3.2036 & 14.8310& 14.8310\\ \cline{2-6}
		& $\lambda=1.5, \ \Theta \sim E(1,2)$& 2.0417&2.0417 & 6.0035& 6.0035 \\ \cline{2-6}
		& $\lambda=1.5, \ \Theta \sim E(1,4)$ & 2.3830 &2.3830 & 7.8910 &7.8910 \\ \hline
		$PH/M/1/1^{\ast}$ & $\mu=0.5, \ \Lambda \sim E(1,2)$ & 2.7500 &2.7500 & 12.0000&12.0000 \\ \cline{2-6}
		&  $\mu=1.5, \ \Lambda \sim E(1,2)$ & 1.4167 &1.4167& 2.8889&2.8889 \\ \cline{2-6}
		&  $\mu=0.5, \ \Lambda \sim E(1,4)$& 2.6250 &2.6250 & 11.1250 &11.1250 \\ \cline{2-6}
		&  $\mu=1.5, \ \Lambda \sim E(1,4)$  & 1.2917 &1.2917 & 2.3472 & 2.3472\\ \hline
		$M/PH/1/2^{\ast} $ & $\lambda=0.5, \ \Theta \sim E(1,2)$& 3.1089 &3.1089 & &  \\ \cline{2-4}
		& $\lambda=0.5, \ \Theta \sim E(1,4)$& 3.0786 &3.0786 &  &   \\ \cline{2-4}
		&  $\lambda=1.5, \ \Theta \sim E(1,2)$ & 2.0996 & 2.0996 &  &  \\ \cline{2-4}
		& $\lambda=1.5, \ \Theta \sim E(1,4)$ & 2.0226 & 2.0226 &  &   \\ \cline{1-4}
	\end{tabular}
	\label{table2}
\end{table}

As a second example, we compare in Fig.~\ref{fig:fig9} our proposed analytical method for the $PH/D/1/1$ queue with $\Lambda \sim E(\lambda^{-1},2)$ against
the upper bound proposed by \cite{champati_etal_infocom19} for the age violation probability $G_{\Delta}(x)= 1 - F_{\Delta}(x), \ x\geq 0$. For our method to work, we approximate the deterministic service time by the Erlang distribution $E(\mu^{-1},j)$ with two relatively large values for the order parameter $j=10,100$.
We use the same numerical example of \cite{champati_etal_infocom19} for which in Fig.~\ref{fig:fig9}(a), $G_D(5)$ is plotted with respect $\lambda$ when $\mu=1$ and in Fig.~\ref{fig:fig9}(b),
$G_D(x)$ is plotted with respect to the age limit $x$ when $\lambda=0.45$ and $\mu=1$.
Fig.~\ref{fig:fig9} demonstrates that the Erlang approximation with order $j=100$ employed along with the proposed method
is quite successful, whereas the easy to obtain upper bound is quite loose for this particular example as shown also in \cite{champati_etal_infocom19}.

%
\begin{figure}[t]
	\centering
	\includegraphics[width=0.85\linewidth]{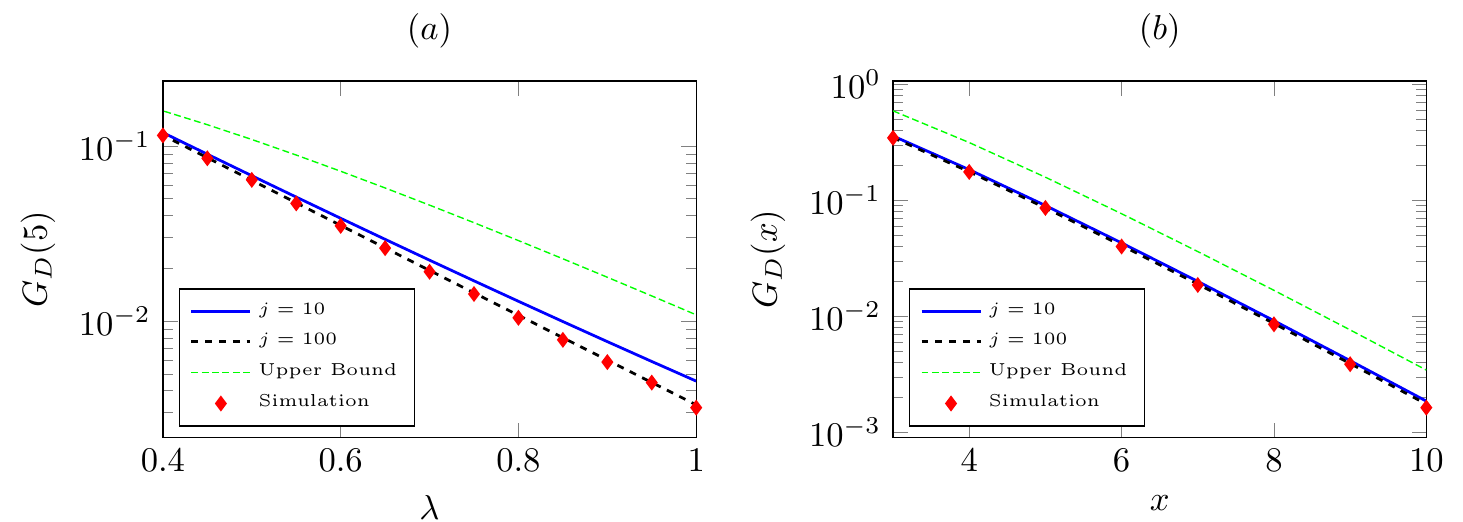}
	\caption{The age violation probability $G_{\Delta}(x)$ (a) with respect to the arrival rate $\lambda$ when $x=5, \ \mu=1$ and 
		(b) with respect to the age limit $x$ when $\lambda=0.45, \ \mu=1$,	obtained by simulations and the proposed method along with upper bound of
		\cite{champati_etal_infocom19}.}
	\label{fig:fig9}
\end{figure}
\subsection{Analytical Results}
In this subsection, only results obtained with the proposed method are reported for the four queueing systems, namely $PH/PH/1/1$, $PH/PH/1/1^{\ast}$, $M/PH/1/2$, $M/PH/1/2^{\ast}$ queues, corresponding to the choices of the preemption parameter $p=0$, $p=1$ and replacement parameter $r=0$, $r=1$, respectively. As the first example of this subsection, the mean AoI and the mean PAoI are 
depicted in Fig.~\ref{fig:fig10} as a function of the scov of the service time for the four queueing models of interest for three different values of the system load, namely $\rho=0.5,1,1.5$, representative of light, critical, and high system load scenarios, respectively. We have the following observations:
\begin{itemize}
	\item The $M/PH/1/1^{\ast}$ model outperforms all the other models in terms of minimizing both $E[\Delta]$ and $E[\Phi]$ for larger values of $c_{\Theta}^2$. However, the $M/PH/1/1^{\ast}$ model performs quite poorly for smaller values of $c_{\Theta}^2$, a situation which is emphasized even more, for larger system loads. 
	This is not surprising since when $c_{\Theta}^2 < 1$, the expected residual service time of the ongoing service time will tend to be smaller than $E[\Theta]$ and the benefit of preemption is to diminish in such situations. 
	\item The PAoI metric $E[\Phi]$ in the $M/PH/1/1$ model exhibits insensitivity to the scov of the service time which is resemblant of the blocking probability in an $M/G/c/c$ system which is dependent on the service time only through its mean and not its higher moments. However, this feature is not inherited in the mean AoI  which appears to increase with increased $c_{\Theta}^2$. 
	\item Both the AoI and PAoI metrics of interest increase with increased values of  $c_{\Theta}^2$ for the single-buffer $M/PH/1/2$ and $M/PH/1/2^{\ast}$ queues for three values of the system load that we studied, whereas the $M/PH/1/2^{\ast}$ system presents consistently better performance since a new information packet is always timelier at the remote server than the one already waiting in the queue.
\end{itemize}
\begin{figure}[t]
	\centering
	\includegraphics[width=1\linewidth]{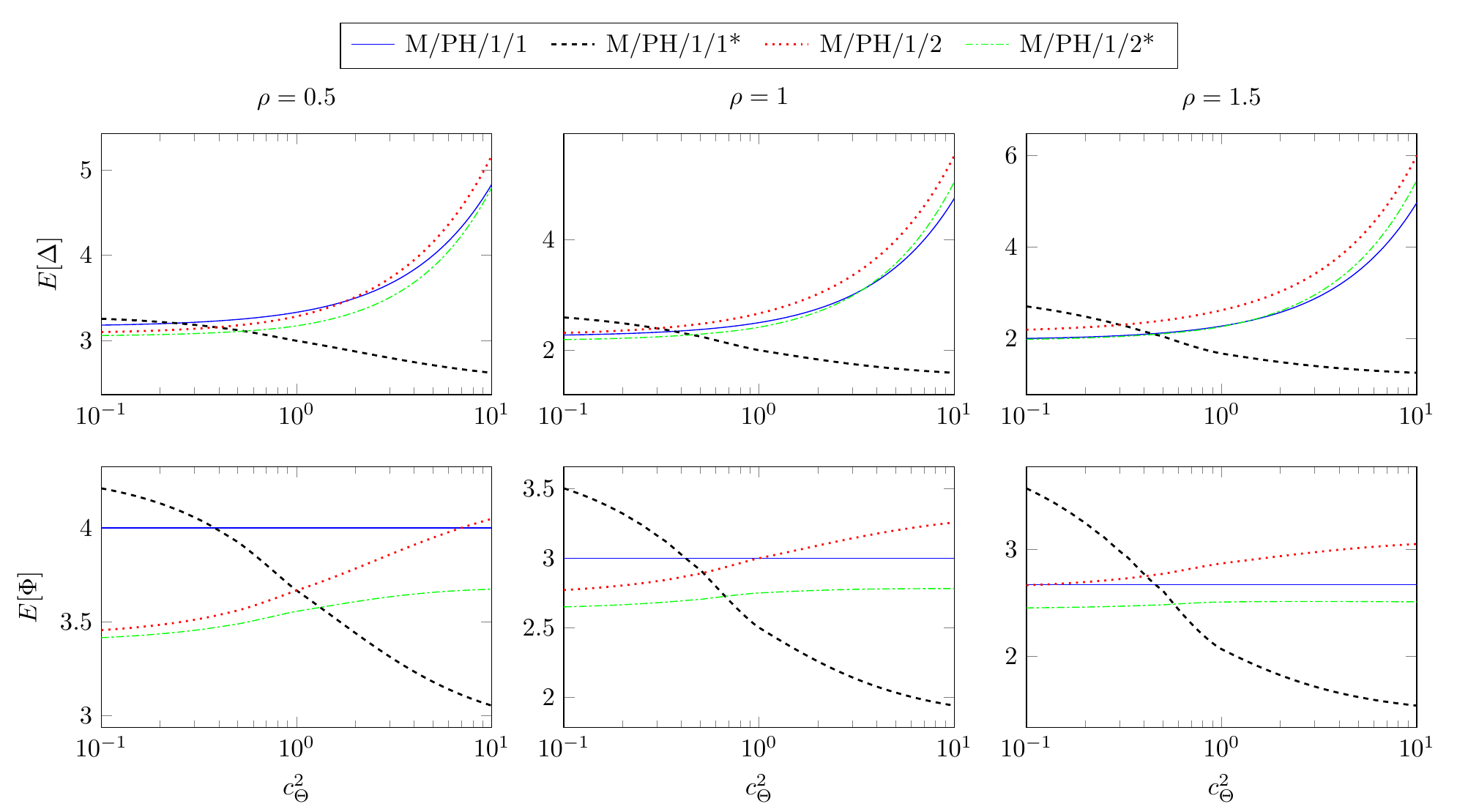}
	\caption{The mean AoI and the mean PAoI as a function of the scov of the
		service time for the four proposed queueing models for three 
		different values of $\rho=0.5,1,1.5$.}
	\label{fig:fig10}
\end{figure} \par
In the second example, the queueing model (out of four) 
which gives rise to the minimum mean AoI is depicted in 
Fig.~\ref{fig:fig11}(a) as a function of the system load $\rho$ and 
scov of the service time $c_{\Theta}^2$. For larger values of 
$c_{\Theta}^2$, the best system is the $M/PH/1/1^{\ast}$ model 
irrespective of the system load. However, for smaller values of 
$c_{\Theta}^2$, the $M/PH/1/1$ model gives the best performance 
for larger values of the system load, whereas it is taken down by the 
$M/PH/1/2^{\ast}$ model for lower values of the system load. 
In some cases, preemption may not be possible since the information 
packet in service may not be under the control of the server once the 
service begins. Consequently, we depict the best queueing model out of 
three models only (when the preemptive $M/PH/1/1^{\ast}$ model is 
excluded) in Fig.~\ref{fig:fig11}(b) which shows that the boundary 
between the queueing models 
$M/PH/1/2^{\ast}$ and $M/PH/1/1$ turn out to depend on the particular value of $c_{\Theta}^2$ when $c_{\Theta}^2 >1$. The $M/PH/1/2$ model does not give rise to the best mean AoI figure in any of these plots since it is always outperformed by $M/PH/1/2^{\ast}$.
\begin{figure}[tb]
	\centering
	\includegraphics[width=0.8\linewidth]{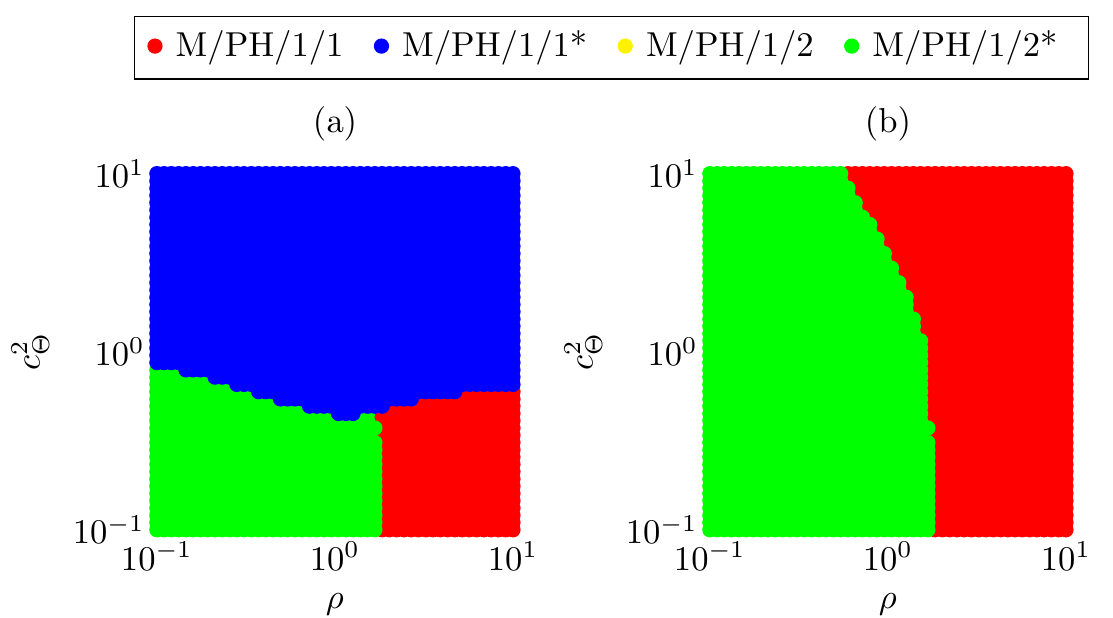}
	\caption{The queueing model which minimizes the mean AoI as a function of
		the system load $\rho$ and scov of the service time $c_{\Theta}^2$ (a) 
		out of all four queueing models (b) out of three queueing models when 
		$M/PH/1/1^{\ast}$ is excluded.}
	\label{fig:fig11}
\end{figure} \par
As a final example, we study the impact of the scov of the interarrival times in the mean AoI. Fig.~\ref{fig:fig12} depicts the mean AoI with respect to varying $c_{\Lambda}^2$
for four different values of  $c_{\Theta}^2$ for a critically loaded system with $\rho=1$.
Both queueing models tend to be affected adversely with increased $c_{\Lambda}^2$, but for larger values of $c_{\Lambda}^2$, the $PH/PH/1/1$ system is penalized more severely.
\begin{figure}[tb]
	\centering
	\includegraphics[width=1\linewidth]{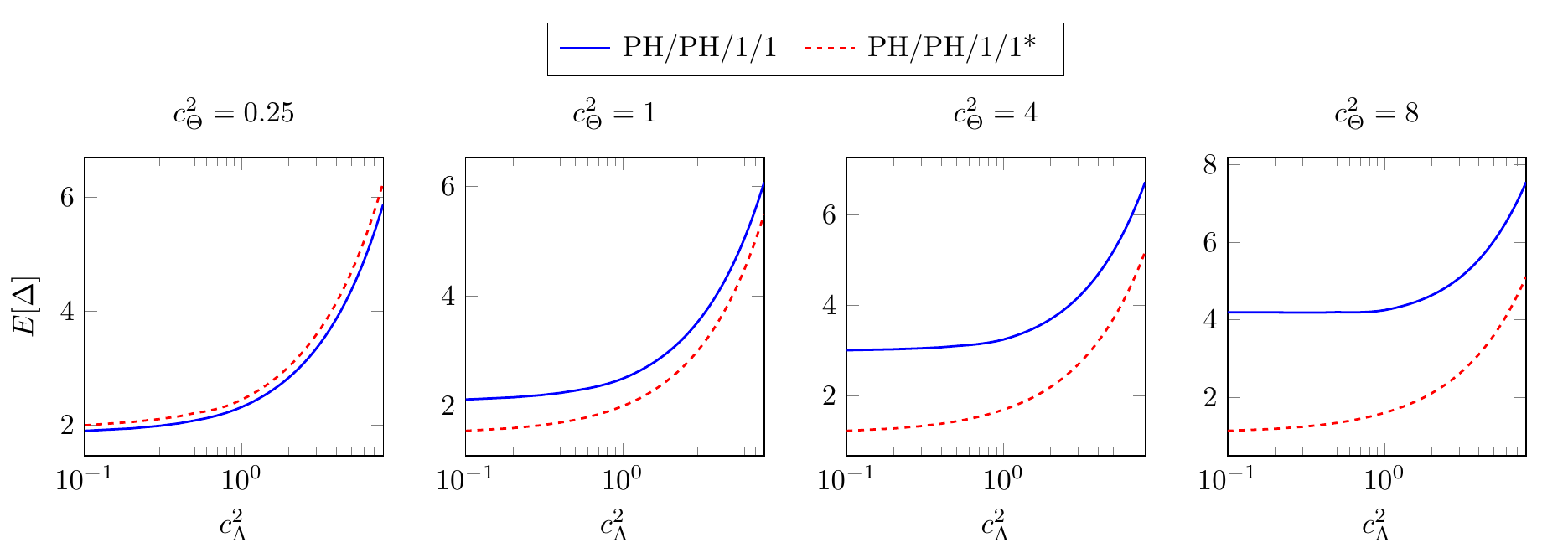}
	\caption{Mean AoI for the $PH/PH/1/1$ and $PH/PH/1/1^{\ast}$ queues with respect to varying $c_{\Lambda}^2$
		for four different values of  $c_{\Theta}^2$ when $\rho=1$.}
	\label{fig:fig12}
\end{figure}

\section{Conclusions}
In this paper, we propose and validate a computationally efficient and stable numerical method for obtaining the exact steady-state distributions of both the AoI and PAoI processes for the bufferless $PH/PH/1/1/P(p)$ queue with probabilistic preemption and the $M/PH/1/2/R(r)$ queue with probabilistic replacement. Age violation probabilities and the associated moments can be calculated easily since the obtained distributions are in matrix exponential form. Cases for which probabilistic preemption is beneficial for bufferless systems with respect to complete preemption and no preemption cases are identified. For single buffer systems, the complete replacement policy is shown to yield the best AoI performance. 
The scov of the service time is shown to play an important role for all the queueing systems of interest on the mean values of AoI and PAoI with the exception that the mean PAoI for the $M/PH/1/1$ system appears to be insensitive to higher order moments of the service time. 
Future work will include extensions of the proposed model to update systems involving multiple sources.

\end{document}